%% file: rnn.tex
\documentclass[11pt]{article}

\usepackage[english]{babel}
\usepackage{amsmath, amsthm, amssymb}
\usepackage[boxed,lined]{algorithm2e}
\usepackage{algpseudocode}
\usepackage{nicefrac}
\usepackage{enumerate}
\usepackage{comment}
\usepackage{graphics}
\usepackage[dvipdf]{graphicx}
\usepackage[numbers,square]{natbib}
\usepackage{xspace} 
\usepackage{fullpage}

\input{header}
\title{Reverse Nearest Neighbors Search in High Dimensions\\using Locality-Sensitive Hashing}
\author{\begin{tabular}{ccc}
David Arthur &\hspace{50pt}& Steve Y. Oudot\\
Department of Computer Science && Geometrica group\\
Stanford University && INRIA Saclay -- \^Ile-de-France\\
{\tt darthur@gmail.com} && {\tt steve.oudot@inria.fr}
\end{tabular}}

\begin{document}
\maketitle

\input{abstract}

\input{core}

\bibliographystyle{plainnat}
\bibliography{references}

\end{document}

%% file: header.tex
\newtheorem{thm}{Theorem}[section]
\newtheorem{lem}[thm]{Lemma}
\newtheorem{cor}[thm]{Corollary}
\newtheorem{defn}[thm]{Definition}

\newtheorem{remark}[thm]{Remark}
\newtheorem{problem}{Problem}
\newtheorem{observation}{Observation}

 \newenvironment{slist}%
 { \begin{list}%
         {--}%
         {\setlength{\labelwidth}{20pt}%
          \setlength{\leftmargin}{25pt}%
          \setlength{\topsep}{0pt}
          \setlength{\itemsep}{0pt}
          \setlength{\parsep}{0pt}}}%
 { \end{list} }

\newcommand{\R}{\mathbb{R}}
\newcommand{\Z}{\mathbb{Z}}

\newcommand{\Prob}{\Phi}
\newcommand{\p}[1]{\Pr\{#1\}}

\newcommand{\cdf}{F}
\newcommand{\cdfn}{\cdf_\mathcal{N}}

\newcommand{\dist}{\bd}

\newcommand{\plog}{\textup{polylog}}

\newcommand{\e}{\varepsilon}
\newcommand{\bd}{\mathrm{d}}

\newcommand{\f}{\mathcal{F}}
\newcommand{\g}{\mathcal{G}}

\newcommand{\T}{\mathcal{T}}

\newcommand{\A}{\mathcal{A}}

\newcommand{\ball}{\ensuremath{\mathcal{B}}\xspace}
\newcommand{\nn}{\ensuremath{\mathcal{NN}}\xspace}
\newcommand{\pleb}{\ensuremath{\mathcal{PLEB}}\xspace}
\newcommand{\rnn}{\ensuremath{\mathcal{RNN}}\xspace}

\newcommand{\nnpq}[2]{\nn_{#1}(#2)}
\newcommand{\nnp}[1]{\nnpq{P}{#1}}
\newcommand{\ennpq}[3]{\ensuremath{#1\mbox{-}}\nnpq{#2}{#3}}
\newcommand{\ennp}[2]{\ennpq{#1}{P}{#2}}

\newcommand{\nnpqr}[3]{\ball_{#1}(#2, #3)}
\newcommand{\nnpr}[2]{\nnpqr{P}{#1}{#2}}

\newcommand{\rnnpq}[2]{\rnn_{#1}(#2)}
\newcommand{\rnnp}[1]{\rnnpq{P}{#1}}
\newcommand{\ernnpq}[3]{\ensuremath{#1\mbox{-}}\rnnpq{#2}{#3}}
\newcommand{\ernnp}[2]{\ernnpq{#1}{P}{#2}}
\newcommand{\rnnds}{\ensuremath{\mathcal{RNNDS}}}

\newcommand{\dnn}[2]{\dist(#2, #1)}

\newcommand{\dnnp}[1]{\dnn{P}{#1}}
\newcommand{\dnnpp}[1]{\dnnp{#1}}

%% file: abstract.tex
\begin{abstract}
We investigate the problem of finding reverse nearest neighbors
efficiently. Although provably good solutions exist for this problem
in low or fixed dimensions, to this date the methods proposed in high
dimensions are mostly heuristic. We introduce a method that is both
provably correct and efficient in all dimensions, based on a
reduction of the problem to one instance of $\e$-nearest neighbor
search plus a controlled number of instances of {\em exhaustive
  $r$-\pleb}, a variant of {\em Point Location among Equal Balls}
where all the $r$-balls centered at the data points that contain the query
point are sought for, not just one. The former problem has been extensively
studied and elegantly solved in high dimensions using
Locality-Sensitive Hashing (LSH) techniques. By contrast, the latter
problem has a complexity that is still not fully understood. We revisit the
analysis of the LSH scheme for exhaustive $r$-\pleb using a
somewhat refined notion of locality-sensitive family of hash function,
which brings out a meaningful output-sensitive term in the complexity
of the problem. Our analysis, combined with a non-isometric lifting of
the data, enables us to answer exhaustive $r$-\pleb queries (and
down the road reverse nearest neighbors queries) efficiently. Along
the way, we obtain a simple algorithm for answering exact nearest
neighbor queries, whose complexity is parametrized by some {\em
  condition number} measuring the inherent difficulty of a given
instance of the problem.
\end{abstract}

%% file: core.tex
\section{Introduction}

Proximity queries are ubiquitous in science and engineering, and given
their natural importance they have received a lot of attention from
the computer science community~\citep{Clarkson99, Clarkson06, Indyk04,
  SDI05}. {\em Nearest Neighbor} (\nn) search is certainly among the
most popular ones. Given a finite set $P$ with $n$ points sitting in
some metric space $(X,\bd)$, the goal is to preprocess $P$ in such a
way that, for any query point $q\in X$, a nearest neighbor of $q$
among the set $P\setminus\{q\}$ can be found quickly.  The \nn query
can be easily answered in linear time by brute force search, so the
algorithmic challenge is to preprocess the data points so as to find
the answer in sub-linear time. Numerous methods have been proposed,
however their performances degrade significantly when the
dimensionality $d$ of the data increases --- a phenomenon known as the
{\em curse of dimensionality}.  Typically, they suffer
from either space or query time that is exponential in $d$ , and so
they become no better than brute-force search when $d$ becomes higher
than a few dozens or hundreds~\cite{WSB98}.

In light of the apparent hardness of \nn search, an approximate
version of the problem called $\e$-\nn has been considered, where the
answer can be any point of $P\setminus\{q\}$ whose distance to $q$ is
within a given factor $(1+\e)$ of the true nearest neighbor
distance~\citep{AMNSW98, Clarkson94, IM98, Kleinberg97,
  KOR98}. Inspired from the random projection techniques developed by
\citet{Kleinberg97}, \citet{IM98} and \citet{KOR98} proposed data
structures to answer $\e$-\nn queries with truly sublinear runtime and
fully polynomial space complexity. The approach developped
in~\cite{IM98} is based on the idea of {\em Locality-Sensitive
  Hashing} (LSH), which consists in hashing the data and query points
into a collection of tables indexed by random hash functions, such
that the query point $q$ has more chance to collide with nearby data
points than with data points lying far away. This technique solves a
decision version of the $\e$-\nn problem called {\em Point Location
  among Equal Balls} ($(r,\e)$-\pleb), which asks to decide whether
the distance of $q$ to $P\setminus\{q\}$ is below a given threshold
$r$ or above $r(1+\e)$.  The output is proven correct with high
probability, and the query time is bounded by $O(dn^\varrho\plog\;n)$
for some constant $\varrho = \frac{1}{1+\Theta(\e)}$. Moreover,
\citet{IM98} proposed a reduction of $\e$-\nn search to a
poly-logarithmic number of $(r,\e)$-\pleb queries, thus providing a
fully sublinear-time and polynomial-space procedure for solving
$\e$-\nn. Although originally designed for the Hamming cube, LSH was
later extended~\cite{AI06, DIIM04, HarPeledIndykMotwani} to affine
spaces $\R^d$ equipped with $\ell_s$-norms, $s\in (0,2]$.

In this paper we mainly focus on the reverse problem, known as {\em
  Reverse Nearest Neighbors} (\rnn) search. Given a finite set $P$
with $n$ points sitting in some metric space $(X,d)$, the goal is to
preprocess $P$ in such a way that, for any query point $q\in X$, one
can find the {\em influence set} of $q$, i.e. the set $\rnnpq{P}{q}$
formed by the points $p\in P\setminus\{q\}$ that are closer to $q$
than to $P\setminus \{p\}$. Such points are called {\em reverse
  nearest neighbors} of $q$.  \rnn queries arise in many different
contexts, and it is no surprise that they have
received a lot of attention since their formal introduction
by~\citet{KM00}. A wealth of methods have been
proposed~\citep{ABKKPR06,BJKS06,Clarkson03,FP09,
  KMSXZ07,KM00,KJG08,SFT03,SAA00,TPL04,TYM06}, which behave well in
practice on some classes of inputs. However, these methods are mostly
heuristic, and to date very little is known about the theoretical
complexity of \rnn search, except in low~\citep{CDLSV09,MVZ02} or
fixed~\citep{CVY09} dimensions, where the dimensionality of the data
can be considered as a mere constant. The crux of the matter is that,
in contrast to ($\e$-)\nn search, the answer to an \rnn query is not a
single point but a set of points, whose size can be up to exponential
in the ambient dimension \citep{PZ04}, so there is no way to achieve a
systematic sub-linear query time. Ideally, one would like to achieve a
query time of the form $\tilde O(n^\varrho + |\rnnpq{P}{q}|)$, where
$\varrho$ is a constant less than $1$ and $|\rnnpq{P}{q}|$ is the size
of the reverse nearest neighbors set. The big-$\tilde O$ notation may
hide extra factors that are polynomial in $d$ and
poly-logarithmic in $n$.  Intuitively, the first term in the bound
would represent the incompressible time needed to locate the query
point $q$ with respect to the point cloud $P$, as in a standard \nn
query, while the second term would represent the size of the
sought-for answer.

\paragraph*{Our contributions.}
Our main contribution (see Section~\ref{sec:rnn}) is a reduction of
\rnn search to one instance of $\e$-\nn search plus a poly-logarithmic
number of instances of {\em exhaustive $r$-\pleb}, 
a set-theoretic version of \pleb where not only one $r$-ball
containing the query point $q$ is sought for, but all such balls. Our
reduction is based on a partitioning of the data points into buckets
according to their nearest neighbor distances, combined with a pruning
strategy that prevents the inspection of too many buckets at query
time.

Turning our reduction into an effective algorithm for \rnn search
requires to adapt the LSH scheme to solve exhaustive $r$-\pleb
queries. Such an adatptation was proposed in~\cite[Chapter~1]{SDI05},
with expected query time $\tilde O(n^{\varrho}
+n^\varrho|\nnpqr{P}{q}{r(1+\e)}|)$, where $\varrho =
\frac{1}{1+\Theta(\e)}$ and where $\e>0$ is a user-defined
parameter. Even though the ouput of the query is the set
$\nnpqr{P}{q}{r}$, the query time depends on the size of the superset
$\nnpqr{P}{q}{r(1+\e)}$, and when choosing $\e$ the user must find a
trade-off between increasing the size of $\nnpqr{P}{q}{r(1+\e)}$ and
increasing the average retrieval cost $n^\varrho$ per point of
$\nnpqr{P}{q}{r(1+\e)}$. In Section~\ref{sec:all-re-nn} we revisit the
analysis of~\cite[Chapter~1]{SDI05} using a somewhat finer concept of
locality-sensitive hashing (see
Definition~\ref{def:stronglshashfunc}), which enables us to quantify
more precisely the amount of collisions with the query point that may
occur within the hash tables stored in the LSH data structure. Taking
advantage of this refined analysis, we propose a simple extra
preprocessing step that reduces the average retrieval cost per point
of $\nnpqr{P}{q}{r(1+\e)}$ down to $n^\alpha$ for some constant
$\alpha\leq \e\varrho<\e$, thereby making the previous trade-off no
longer necessary. The price to pay is a slight degradation of the
absolute term $n^\varrho$ in the complexity bound, which rises to
$n^{\varrho'}$ where $\varrho'=\frac{1}{1+\Theta(\e^2)}$
(Theorem~\ref{thm:allnn_Rd}). All in all, the query time bound becomes
$\tilde O(n^{\varrho'} + n^\alpha|\nnpqr{P}{q}{r(1+\e)}|)$ and
therefore remains sublinear in $n$ as long as $|\nnpqr{P}{q}{r(1+\e)}|\leq
n^{1-\e}$.  Intuitively, our extra preprocessing step consists
in lifting the point cloud $P$ and query point $q$ one dimension
higher through some highly non-isometric embedding, so that the
induced metric distortion moves $q$ away from $P$ and further
concentrates the distribution of the distances to $q$ around the
parameter value $r$, thereby reducing the total number of collisions
with $q$ within the hash tables. The output of the query can still be
proven correct thanks to the fact that the embedding preserves the
order of the distances to $q$. This approach stands in contrast to the
general trend of applying low-distortion embeddings to solve proximity
queries.

Down the road, these advances lead to an algorithm for solving \rnn
queries with high probability in expected $\tilde
O(\frac{1}{\e}n^{1/(1+\Theta(\e^2))} + n^\e
|\ernnpq{O(\e)}{P}{q}|)$ time using fully polynomial space, where $\e>0$
is a user-defined parameter and $\ernnpq{O(\e)}{P}{q}$ is a superset
of $\rnnpq{P}{q}$ whose points are $O(\e)$-close to being true
reverse nearest neighbors of $q$ (Theorem~\ref{thm:rnnquerytime}). To
the best of our knowledge, this is the first algorithm for answering
\rnn queries that is provably correct and efficient in
all dimensions. Furthermore, the algorithm and its analysis extend
naturally to the bichromatic setting where the data points are split
into two disjoint categories, e.g. clients and servers, a scenario
that is encountered in various applications~\cite{KM00}.

Along the way, in Section~\ref{sec:exact-nn} we obtain a simple
algorithm that can answer exact \nn queries in expected $\tilde
O(n^{1/(1+\Theta(\e^2))} + n^\e |\ennpq{O(\e)}{P}{q}|)$ time using
fully polynomial space, where $\e>0$ is a user-defined parameter and
$\ennpq{O(\e)}{P}{q}$ is a set of approximate nearest neighbors of $q$
(Theorem~\ref{thm:exact-nn}). The first term in the running time bound
corresponds to a standard $\e$-\nn query, while the second term is
parametrized by the size of $\ennpq{O(\e)}{P}{q}$, which thereby plays
the role of a {\em condition number} measuring the discrepancy in
difficulty between the exact and approximate \nn queries on a given
instance. Note that our algorithm is not expected to perform as well
as state-of-the-art techniques in growth-restricted
spaces~\cite{Clarkson99, hkmr-nnngrm-04, kr-fnngrm-02, kl-ns-04},
however its complexity bounds hold in a more general setting and its
sublinear behavior on a particular instance relies on the weaker
hypothesis that the condition number of this instance lies below the
threshold $n^{1-\e}$. In the same spirit, \citet{DIIM04} designed a
lightweight version of our algorithm that only works in Euclidean
spaces but is competitive with~\cite{Clarkson99, hkmr-nnngrm-04,
  kr-fnngrm-02, kl-ns-04}.

Throughout the paper, the analysis is carried out either in full
generality in metric spaces that admit locality-sensitive families of
hash functions, or more precisely in $(\R^d, \ell_s)$ when liftings of
the data one dimension higher come into play.  The case of the
$d$-dimensional Hamming cube is also encompassed by our analysis since
this space embeds itself isometrically into $(\R^d, \ell_1)$.

\section{Preliminaries}\label{sec:prelim}

In Section~\ref{sec:defs} we introduce some useful notation and state
the nearest neighbor and reverse nearest neighbors problems
formally. In Sections~\ref{sec:renn-to-enn}
through~\ref{sec:pleb_affine} we give an overview of LSH and its
application to approximate nearest neighbor search, with a special
emphasis on the case of affine spaces $\R^d$ equipped with
$\ell_s$-norms in Section~\ref{sec:pleb_affine}. The data structures
and algorithms introduced in this section are used as black-boxes in
the rest of the paper.

\subsection{Problem statements and notations}
\label{sec:defs}

Throughout the paper, $(X,\bd)$ denotes a metric space and $P$ a
finite subset of $X$.  Given a point $x\in X$, let $\dist(x,P)$ denote
the distance of $x$ to $P\setminus\{x\}$, that is: $\dist(x,P) = \min
\left\{\dist(x,p) \mid p\in P\setminus\{x\}\right\}.$ Given a
parameter $r\geq 0$, let $\ball(x,r)$ denote the metric ball of
center $x$ and radius $r$, and let $\nnpqr{P}{x}{r}$ be the set of
points of $P\setminus\{x\}$ that lie within this ball.  Then,
$\nnpqr{P}{x}{\dist(x,P)}$ is the set of {\em nearest neighbors} of
$x$ among $P\setminus\{x\}$, noted $\nnp{x}$. By analogy, given a
parameter $\e> 0$, $\ennp\e{x}$ denotes the set
$\nnpqr{P}{x}{(1+\e)\dist(x,P)}$ of {\em $\e$-nearest neigbors} of $x$
among $P\setminus\{x\}$. The usual convention is that point $x$ itself
is excluded from these sets, which is not mentioned explicitly in our
notations for simplicity but will be admitted implicitly throughout the paper.
\begin{problem}[\nn] \label{problem:nn}
  Given a query point $q\in X$, the {\em nearest neighbor query} asks
  to return any point of $\nnp{q}$.
\end{problem}
\begin{problem}[$\e$-\nn] \label{problem:enn}
  Given a query point $q\in X$, the {\em $\e$-nearest neighbor query}
  asks to return any point of $\ennp\e{q}$.
\end{problem}
Given now a point $x\in X$, let $\rnnp{x}$ denote the set of {\em
  reverse nearest neighbors} of $x$ among $P\setminus\{x\}$, which by
definition are the points $p\in P\setminus\{x\}$ such that $x\in
\nnpq{P\cup\{x\}}{p}$. By analogy, let $\ernnp\e{x}$ denote the set of
     {\em reverse $\e$-nearest neighbors} of $x$ among
     $P\setminus\{x\}$, which by definition are the points $p\in
     P\setminus\{x\}$ such that $x\in \ennpq\e{P\cup\{x\}}{p}$. Here
     again, point $x$ itself is excluded from the various sets, a fact
     omitted in our notations for simplicity but admitted
     implicitly.
\begin{problem}[\rnn] \label{problem:rnn}
  Given a query point $q\in X$, the {\em reverse nearest neighbors
    query} asks to retrieve the set $\rnnp{q}$.
\end{problem}
%

\subsection{Reducing approximate nearest neighbor search to its decision version}
\label{sec:renn-to-enn}

Given a parameter $r$, the decision version of
Problem~\ref{problem:nn} consists in deciding whether $\dist(q,P)$ is
smaller or larger than $r$. This problem is also known as {\em Point
  Location among Equal $r$-Balls} ($r$-\pleb) in the literature, because
it is equivalent to deciding whether $q$ lies inside the union of
balls of same radius $r$ about the points of $P$. It is
formalized as follows:
\begin{problem}[$r$-\pleb] \label{problem:pleb}
  Given a query point $q\in X$, the {\em $r$-\pleb} query asks
  the following:
  \begin{slist}
  \item[$\bullet$] if $\dist(q,P)\leq r$, then return {\em
    YES} and any point $p\in P$ such that $\dist(p,q)\leq r$;
  \item[$\bullet$] else {\rm ($\dist(q,P)> r$)}, return {\em NO}.
  \end{slist}
\end{problem}
By analogy, the decision version of Problem~\ref{problem:enn} consists
in deciding whether $\dist(q,P)$ is smaller than $r$ or larger than
$r(1+\e)$. If it lies between these two bounds, then any answer is
acceptable.  The formal statement is the following:
\begin{problem}[$(r,\e)$-\pleb] \label{problem:epleb}
  Given a query point $q\in X$, the
  {\em $(r,\e)$-\pleb}  query asks the following:
  \begin{slist}
  \item[$\bullet$] if $\dist(q,P)\leq r$, then return {\em
    YES} and any point $p\in P$ such that $\dist(p,q)\leq r(1+\e)$;
  \item[$\bullet$] if $\dist(q,P)> r(1+\e)$, then return {\em NO};
  \item[$\bullet$] else {\rm ($r < \dist(q, P)\leq r(1+\e)$)},
    return any of the above answers.
  \end{slist}
\end{problem}

The original LSH paper \citep{IM98} showed a construction that reduces
the $\e$-\nn problem to a logarithmic number of $(r,\e)$-\pleb
queries. Other reductions have since been proposed, and in this paper
we will make use of the following one, introduced by \citet{HarPeled01},
which is simple and works in any metric space. It is based on
a divide-and-conquer strategy, building a tree $\T(P,\e)$ of height
$O(\ln n)$, such that each node $v$ is assigned a subset
$P_v\subseteq P$ and an interval $[r_v, R_v]$ of possible values for
parameter $r$.  Each $\e$-\nn query is performed by traversing down
the search tree $\T(P,\e)$, and by answering two $(r, \e)$-\pleb
queries at each node $v$ to decide (approximately) whether
$\dist(q,P)$ belongs to the interval $[r_v, R_v]$ or not: in the
former case, a simple dichotomy on a geometric progression of values
of $r$ within the interval makes it possible to determine within a
relative error of $1+\e$ where $\dnn{P}{q}$ lies in the interval, and
to return a point of $\ennpq\e{P}{q}$, with a total number of $(r,
\e)$-\pleb queries bounded by $O(\log_2 \log_{1+\e} \frac{R_v}{r_v})$;
in the latter case, the choice of the child of $v$ in which to
continue the search is determined from the output of the two $(r,
\e)$-\pleb queries. In this construction, the ratio $\frac{R_v}{r_v}$
is guaranteed to be at most a polynomial in $\frac{n}{\e}$, with
bounded degree, so we have $O(\log_{1+\e}\frac{R_v}{r_v}) =
O(\log_{1+\e}\frac{n}{\e})= O(\frac{1}{\e}\ln \frac{n}{\e})$.  Thus,
\begin{thm}[see~\cite{HarPeled01}] \label{thm:cnn} 
Given a finite set $P\subseteq X$ with $n$ points, the tree $\T(P,\e)$
stores $O(\frac{1}{\e}\ln \frac{n}{\e})$ data structures for
$(r,\e)$-\pleb queries per node, and it reduces every $\e$-\nn query
to a set of $O(\ln n + \ln
\frac{1}{\e}+\ln\ln\frac{n}{\e})=O(\ln \frac{n}{\e})$ queries of type
$(r,\e)$-\pleb.
 \end{thm}

\subsection{Solving $(r,\e)$-\pleb queries by means of Locality-Sensitive Hashing}
\label{sec:lsh-ernn}

\begin{defn}\label{def:lshashfunc}
  Given a metric space $(X,\dist)$ and two radii $r_1<r_2$, a family
  $\f=\{f:X\rightarrow \Z\}$ of hash functions is called {\em
    $(r_1,r_2,p_1,p_2)$-sensitive} if there exist quantities $1> p_1 >
  p_2 > 0$ such that $\forall x,y\in X$,
  \begin{slist}
  \item[$\bullet$] $\bd(x,y)\leq r_1$ $\Rightarrow$ $\p{f(x)=f(y)}\geq p_1$,
  \item[$\bullet$] $\bd(x,y)\geq r_2$ $\Rightarrow$ $\p{f(x)=f(y)}\leq p_2$,
  \end{slist}
  where probabilities are given for a random choice of hash
  function $f\in \f$ according to some probability distribution
  over the family.
\end{defn}
Intuitively, a $(r_1,r_2,p_1,p_2)$-sensitive family of hash functions
distinguishes points that are close together from points
that are far apart. 

Assuming that a $(r, r(1+\e), p_1, p_2)$-sensitive family $\f$ of hash
functions is given, it is possible to answer $(r,\e)$-\pleb queries in
sub-linear time \citep{GIM99,IM98}. The algorithm proceeds as follows:
\begin{itemize}
\item In the pre-processing phase, it {\em boosts} the sensitivity of
  the family $\f$ by building $k$-dimensional vectors $g=(f_1, \cdots,
  f_k): X\to \Z^k$ whose coordinate functions $f_i$ are drawn
  independently at random from $\f$. The hash key of a point $x\in X$
  is now a $k$-dimensional vector $g(x)=(f_1(x), \cdots, f_k(x))$, and
  two keys $g(x)$ and $g(y)$ are equal if and only if
  $f_i(x)=f_i(y)$ for all $i=1, \cdots, k$. Call $\g$ the family of
  such random hash vectors. The algorithm draws $L$ elements $g_1,
  \cdots, g_L$ independently from $\g$, and it builds the $L$
  corresponding hash tables $H_1, \cdots, H_L$. It then hashes each
  data point $p\in P$ into every hash table $H_i$ using vector
  $g_i(p)$ as the hash key.

\item In the online query phase, the algorithm hashes the query point
  $q$ into each of the $L$ hash tables, and it collects all the
  points colliding with $q$ therein, until either some point $p\in
  \nnpqr{P}{q}{r(1+\e)}$ has been found or more than $3L$ points
  (including duplicates) have been collected in total. In the former
  case the algorithm answers YES and returns $p$, while in the latter
  case it answers NO. It also answers NO if no point of
  $\nnpqr{P}{q}{r(1+\e)}$ has been found after visiting all the hash
  tables.
\end{itemize}

Letting $k=\lceil\frac{\ln n}{\ln \nicefrac{1}{p_2}}\rceil$ and
$L=\lceil \frac{n^\varrho}{p_1}\rceil$, where $\varrho=\frac{\ln p_1}{\ln p_2}$,
one can prove that this procedure gives the correct answer with
constant probability~\cite{GIM99,IM98}. By repeating it $\omega\ln n$
times, for a fixed constant $\omega> 0$, one can increase the probability of
success to at least $1-\frac{1}{n^\omega}$.  Thus,
\begin{thm}[see ~\citep{GIM99,IM98}]\label{th:lsh4renn}
  Given a finite set $P$ with $n$ points in $(X, \bd)$, two parameters
  $r,\e> 0$, and a $(r,r(1+\e),p_1,p_2)$-sensitive family $\f$ of hash
  functions for some constants $p_1>p_2$, the LSH data structure has
  size $O(\frac{n^{1+\varrho}}{p_1}\ln n)$ and answers $(r,\e)$-\pleb
  queries correctly with high probability in
  $O(\frac{n^\varrho}{p_1}\frac{\ln n}{\ln\nicefrac{1}{p_2}}\ln n)$
  time, where $\varrho = \frac{\ln p_1}{\ln p_2}<1$.
\end{thm}
Note that the running time bound ignores the time needed to compute
distances and to evaluate hash functions. These typically depend on
the metric space $(X,\bd)$ and hash family $\f$ considered.
The probabilities $p_1,p_2$ also depend on
$\f$, therefore they may vary with $r$ and $\e$.

\subsection{The case of affine spaces}
\label{sec:pleb_affine}

In most of the paper the ambient space $X$ will be the affine
space $\R^d$ equipped with some $\ell_s$-norm, $s\in
(0,2]$, and $\dist$ will denote the induced distance: $\forall x,y\in
\R^d$, $\dist(x,y) = \|x-y\|_s = \left(\sum_{i=1}^d
|x_i-y_i|^s\right)^{1/s}$, where $x_i,y_i$ stand for the $i$-th
coordinates of $x,y$.

In $(\R^d, \ell_s)$ we use the families of hash functions introduced
by~\citet{DIIM04}\footnote{A possible improvement would be to use the
  hash functions defined by~\citet{AI06} instead, which are known to
  give better complexity bounds. For now we leave this as future
  work.}, which are derived from so-called {\em $s$-stable
  distributions}. A distribution $D$ over the reals is called
$s$-stable if any linear combination $\sum_{i} \alpha_i X_i$ of
finitely many independent variables $X_i$ with distribution $D$ has
the same distribution as $(\sum_i |\alpha_i|^s)^{1/s} X$, where $X$ is
a random variable with distribution $D$.  Given such a distribution
$D$, one can build $(r, r(1+\e), p_1, p_2)$-sensitive families of hash
functions in $(\R^d, \ell_s)$ for any radius $r> 0$ and any
approximation parameter $\e>0$ as follows. First, rescale the data and
query points so that $r=1$. Then, choose a real value $w>0$ and define
a two-parameters family of hash functions
$\f=\{f_{a,b}:\R^d\to\Z\}_{a\in\R^d, b\in [0,w)}$ by
  $f_{a,b}(x)=\lfloor \frac{x\cdot a + b}{w}\rfloor$, where $\cdot$
  stands for the inner product in $\R^d$. The probability distribution
  over the family is not uniform: the coordinates of vector
  $a$ are chosen independently according to $D$, while $b$ is drawn
  uniformly at random from the interval $[0,w)$.  The local
    sensitivity of this family depends on the choice of parameter
    $w$. More precisely, according to~\citet{DIIM04}, given two points
    at distance $l$ of each other, the probability (over a random
    choice of hash function) that these points collide is
\begin{equation}\label{eq:proba_collision_ls}
\Prob(l)=\int_0^w \frac{1}{l}\; \mathrm{f}_D\!\! \left(\frac{t}{l}\right)
    \left(1-\frac{t}{w}\right)dt,
\end{equation}
 where $\mathrm{f}_D$ denotes the probability density function of the
 absolute value of $D$. The probabilities $p_1,p_2$ in
 Theorem~\ref{th:lsh4renn} are then obtained as $\Prob(1)$ and $\Prob(1+\e)$
 respectively. They do not depend on $r$, thanks to the
 rescaling. Note that they do note depend on the dimension $d$
 either.

Focusing back on Har-Peled's construction, recall from
Theorem~\ref{thm:cnn} that each node $v$ of the tree $\T(P,\e)$ stores
$O(\frac{1}{\e}\ln\frac{n}{\e})$ data structures for answering
$(r,\e)$-\pleb queries, each of size $O(|P_v|^{1+\varrho}\ln
|P_v|)$. Let us point out that by construction the subsets of $P$
assigned to the sons of $v$ form a partition of $P_v$. Then, a
recursion gives the following bounds on the size of $T(P,\e)$ and on
the query time\footnote{Our complexity bounds differ from
  the ones of~\citet{HarPeledIndykMotwani} in that the $\ln \ln n$
  factor in their bounds is replaced by a $\ln n$ factor in ours. This
  difference comes from the fact that we run the LSH procedure
  $\omega\ln n$ times, for a fixed constant $\omega$, to make
  its output correct with probability at least $1-\frac{1}{n^\omega}$,
  so the full $\e$-\nn algorithm can be correct with probability at
  least $1-\frac{1}{n}$, which will be useful in the rest of the
  paper. By contrast, the analysis in~\cite{HarPeledIndykMotwani} only
  runs the LSH procedure $O(\ln \ln n)$ times, to make the $\e$-\nn
  algorithm correct with constant probability.}:
\begin{cor}[see~\cite{HarPeledIndykMotwani}]\label{cor:lsh4enn}
  Given a finite set $P$ with $n$ points in $(\R^d, \ell_s)$, $s\in
  (0,2]$, and a parameter $\e > 0$, the tree structure $\T(P,\e)$ and
  its associated $(r,\e)$-\pleb data structures can answer $\e$-\nn
  queries correctly with high probability in
  $O\left(\frac{n^{\varrho}}{p_1}\frac{\ln n}{\ln \nicefrac{1}{p_2}} \ln n
  \ln\frac{n}{\e}\right)$ time using
  $O\left(\frac{1}{\e}\frac{n^{1+\varrho}}{p_1}\ln^2 n\ln \frac{n}{\e}\right)$
  space, where $\varrho = \frac{\ln p_1}{\ln p_2}<1$, the quantities
  $p_1=\Prob(1)$ and $p_2=\Prob(1+\e)$ being derived from some $s$-stable
  distribution $D$ according to Eq.~(\ref{eq:proba_collision_ls}).
\end{cor}
Here again the running time bound ignores the time needed to compute
distances and to evaluate hash functions, which is $O(d)$ per
operation (distance computation or hash function evaluation) in
$\R^d$. From now on we will also ignore poly-logarithmic factors in
$\frac{n}{\e}$ and hide them within big-$\tilde O$ notations for the
sake of simplicity. Thus, the time and space complexities given in
Theorem~\ref{th:lsh4renn} become respectively $\tilde
O(\frac{n^\varrho}{p_1\ln\nicefrac{1}{p_2}})$ and $\tilde
O(\frac{n^{1+\varrho}}{p_1})$, while those given in Corollary~\ref{cor:lsh4enn}
become respectively $\tilde O(\frac{n^\varrho}{p_1\ln\nicefrac{1}{p_2}})$
and $\tilde O(\frac{1}{\e}\frac{n^{1+\varrho}}{p_1})$.

The challenge now is to choose a value for parameter $w$ that makes
$\varrho$ as small as
possible. The {\em best} value for $w$ heavily depends on $s$ and
$\e$, and it may be difficult to find for some values of $s,\e$,
especially when no closed form solution to
Eq.~(\ref{eq:proba_collision_ls}) is known. Two special cases of
practical interest ($s=1$ and $s=2$) are analyzed
in~\cite{DIIM04}:
\begin{itemize}
\item In the case $s=1$, one can use the Cauchy distribution (which is
  $1$-stable) to derive a family of hash functions, and the
  probability of collision becomes $\Prob(l)=2\frac{\arctan(w/l)}{\pi} -
  \frac{1}{\pi(w/l)}\ln(1+(w/l)^2)$. The ratio $\varrho=\frac{\ln
    p_1}{\ln p_2}$ lies then strictly above $\frac{1}{1+\e}$, yet
  larger and larger values of parameter $w$ make it closer and
  closer to $\frac{1}{1+\e}$.
\item In the case $s=2$, one can use the normal distribution
  $\mathcal{N}(0,1)$ (which is $2$-stable), and the probability of
  collision becomes $\Prob(l) = 1-2\cdfn(-w/l) -
  \frac{2}{\sqrt{2\pi}w/l}(1- e^{-w^2/2l^2})$, where
  $\cdfn$ stands for the cumulative
  distribution function of $\mathcal{N}(0,1)$.  The ratio
  $\varrho=\frac{\ln p_1}{\ln p_2}$ lies then below $\frac{1}{1+\e}$
  for reasonably small values of parameter $w$.
\end{itemize}
The results obtained by~\citet{DIIM04} can be extended to any
$s\in [1,2]$ via low-distortion embeddings~\cite{js-elpilo}.  In the
rest of the paper we will follow~\cite{DIIM04} and use respectively
the Cauchy distribution and the normal distribution in the cases $s=1$
and $s=2$. An analysis of the influence of the choice of
parameter $w$ on the quantities $\varrho$, $\frac{1}{p_1}$ and
$\frac{1}{\ln\nicefrac{1}{p_2}}$ will be provided in
Section~\ref{sec:erpleb_Rd}.

\section{Exhaustive $r$-\pleb}
\label{sec:all-re-nn}

Let $(X,\bd)$ be a metric space and $P$ a finite subset of $X$. The
following variant of $r$-\pleb, where all the $r$-balls containing the
query point are asked to be retrieved, will play a central part in the
rest of the paper:
%
\begin{problem}[Exhaustive $r$-\pleb] \label{problem:allnnexact}
  Given a query point $q\in X$, the {\em exhaustive $r$-\pleb} query
  asks to return the set $\nnpqr{P}{q}{r}$.
\end{problem}
This problem is introduced under the name {\em near-neighbors
  reporting} in previous literature~\cite[Chapter~1]{SDI05}, where a
variant of the LSH scheme of Section~\ref{sec:lsh-ernn} is proposed
for solving it. The difference with the original LSH scheme is that
the query procedure does not stop when $3L$ collisions with the query
point $q$ have been found, but instead it continues until all the
points colliding with $q$ in the $L$ hash tables have been
collected. The output is then the subset of these points that lie
within $\nnpqr{P}{q}{r}$. The details of the pre-processing and query
phases are given in Algorithms~\ref{alg:nnpre} and~\ref{alg:nnquery}
respectively, where the data structure is called $\A(P,r,\e)$. Note
that parameter $\e$ no longer controls the quality of the output,
which is shown to coincide with the set $\nnpqr{P}{q}{r}$ with high
probability, but instead it influences the average complexity of the
procedure, as we will see later on.

\begin{algorithm}[!htb]
  \LinesNumbered
  \SetKwInOut{Input}{Input}
  \SetKwInOut{Output}{Output}

  \Input{metric space $(X,\bd)$, finite set $P$ with $n$ points in $X$, parameters
    $r,\e> 0$}
  \Output{$\A(P, r, \e)$ data structure}
  \BlankLine

  Take an $(r, r(1+\e), p_1, p_2)$-sensitive LSH family $\f$\; 
  Let $k=\lceil\frac{\ln n}{\ln\nicefrac{1}{p_2}}\rceil$ and $L=\lceil \frac{n^\varrho}{p_1}\rceil$, where $\varrho=\frac{\ln p_1}{\ln p_2}$\;
  Create the $k$-dimensional hash family $\g$ as described in Section~\ref{sec:lsh-ernn}\;
  \For{$i=1$ to $\lceil c\ln n\rceil$  \tcp{$c$ is a constant to be explicited later} } {  
    pick $L$ functions $\{g_1,\ldots,g_{L}\}$ independently at random from $\g$\;
    Create the corresponding hash tables $\{H_1,\ldots,H_{L}\}$\;
    \ForAll{$p\in P$} {
      \For{$j=1$ to $L$} {
        Insert $p$ into $H_j$ using the key $g_j(p)$\;
      }
    }
    Store the data structure $\A_i(P, r, \e):=\{g_1, \cdots, g_{L}\}\sqcup \{H_1, \cdots, H_{L}\}$\;
  }
  Output $\A(P, r, \e) := \bigcup_i \A_i(P, r, \e)$\;

  \caption{\em Pre-processing phase for exhaustive $r$-\pleb}
  \label{alg:nnpre}
\end{algorithm}

\begin{algorithm}[!hbt]
  \LinesNumbered
  \SetKwInOut{Input}{Input}
  \SetKwInOut{Output}{Output}

  \Input{metric space $(X,\bd)$, $\A(P, r, \e)$ data structure, query point $q\in X$} 
\BlankLine

  Let $k$, $L$, $\varrho$  and $c$ be defined as in Algorithm~\ref{alg:nnpre}\;
  Initialize the output set: $S:=\emptyset$\;
  \For{$i=1$ to $\lceil c\ln n\rceil$} {
    Let  $\{g_1,\ldots,g_{L}\}$ be the functions and
    $\{H_1,\ldots,H_{L}\}$ the tables contained in $\A_i(P, r,\e)$\;
    \For{$j=1$ to $L$} {
      Compute $g_j(q)$ and retrieve the set $C_j$ of the points colliding with $q$ in $H_j$\;
      \ForAll{$p\in C_j$} {
        \If{$\bd(q,p) \leq r$\nllabel{line:approxnn}} {
          Update the output set:
          $S := S\cup\{p\}$\;
        }
      }
    }\nllabel{line:breakjump}
  }
  Return $S$\;

  \caption{\em Online query phase for exhaustive $r$-\pleb}
  \label{alg:nnquery}
\end{algorithm}

In Section~\ref{sec:erpleb_general} we revisit the analysis
of~\cite[Chapters~1 and~3]{SDI05} and quantify more precisely the
amount of collisions with the query point that may occur within the
hash tables. To this end we use the following refined concept of
locality-sensitive family of hash functions\footnote{An even finer
  concept, proposed in~\cite[\textsection~3.3]{SDI05}, makes the
  probability of having $f(x)=f(y)$ a function of the distance between
  $x$ and $y$. However, for our purpose it is not necessary to go to
  this level of refinement.}:
\begin{defn}\label{def:stronglshashfunc}
  Given a metric space $(X,\dist)$ and positive radii $r_0\leq
  r_1<r_2$, a family $\f=\{f:X\rightarrow \Z\}$ of hash functions is
  called {\em $(r_0,r_1,r_2,p_0,p_1,p_2)$-sensitive} if there exist
  quantities $1>p_0\geq p_1>p_2 > 0$ such that $\forall x,y\in X$,
  \begin{slist}
  \item[\rm (i)] $\bd(x,y)\leq r_1$ $\Rightarrow$ $\p{f(x)=f(y)}\geq p_1$,
  \item[\rm (ii)] $\bd(x,y)\geq r_2$ $\Rightarrow$ $\p{f(x)=f(y)}\leq p_2$,
  \item[\rm (iii)] $\bd(x,y)\geq r_0$ $\Rightarrow$ $\p{f(x)=f(y)}\leq p_0$,
  \end{slist}
  where probabilities are given for a random choice of hash function
  $f\in \f$ according to some probability distribution
  over the family.
\end{defn}
Axioms (i) and (ii) correspond to the classical notion of
locality-sensitive family of hash functions
(Definition~\ref{def:lshashfunc}). They do not make it possible to
limit the number of collisions between the query point $q$ and the
points of $\nnpqr{P}{q}{r_1}$ in the analysis of exhaustive
$r_1$-\pleb queries. Specifically, every point of $\nnpqr{P}{q}{r_1}$
might collide with $q$ in every hash table in theory, thus raising the
cost of an exhaustive $r_1$-\pleb query to $\Omega(n^\varrho)$ per
point of $\nnpqr{P}{q}{r_1}$. This is in fact all theoretical, since
in practice the hash functions are likely to make a difference between
those points of $\nnpqr{P}{q}{r_1}$ that are really close to $q$ and
those that are farther away. This is the reason for introducing the
third axiom (iii), which will prove its usefulness in
Section~\ref{sec:erpleb_Rd}, where we concentrate on the case where
the ambient space is $(\R^d, \ell_s)$, $s\in (0,2]$, and show that a
non-isometric embedding of the data into $(\R^{d+1}, \ell_s)$ enables
us to move the sets of data and query points away from each other.

\subsection{Revisiting the analysis in the general case}
\label{sec:erpleb_general}

\begin{thm} \label{thm:allnn}
  Given a finite set $P\subseteq X$ with $n$ points and two parameters
  $r,\e> 0$, if $(X, \dist)$ admits a $(r_1, r_2, p_1,
  p_2)$-sensitive family $\f$ of hash functions with $r_1=r$ and
  $r_2\leq r(1+\e)$, then Algorithm~\ref{alg:nnquery} answers
  exhaustive $r$-\pleb queries correctly with high probability in
  expected $\tilde O(\frac{n^\varrho}{p_1} (\frac{1}{\ln
    \nicefrac{1}{p_2}}+1+|\nnpqr{P}{q}{r(1+\e)}|) )$ time,
  involving $\tilde O(\frac{n^\varrho}{p_1} +
  |\nnpqr{P}{q}{r(1+\e)}|)$ distance computations and $\tilde
  O(\frac{n^\varrho}{p_1 \ln \nicefrac{1}{p_2}})$ hash function
  evaluations only, and using $\tilde O(\frac{n^{1+\varrho}}{p_1})$ space, where
  $\varrho=\frac{\ln p_1}{\ln p_2}$. If moreover the family $\f$ is
  $(r_0, r_1, r_2, p_0, p_1, p_2)$-sensitive for some $r_0\leq r_1$,
  then for any query point $q\in X$ the algorithm answers the exhaustive
  $r$-\pleb query in expected $\tilde O(\frac{n^\varrho}{p_1}(\frac{1}{\ln
    \nicefrac{1}{p_2}}+1+|\nnpqr{P}{q}{r_0}|) +
  \frac{n^\alpha}{p_1}|\nnpqr{P}{q}{r(1+\e)}\setminus \ball(q, r_0)|)$
  time, where $\alpha=\varrho(1-\frac{\ln p_0}{\ln p_1})\leq\varrho$.
\end{thm}

The first term ($\frac{n^\varrho}{p_1 \ln \nicefrac{1}{p_2}}$) in the
running time bound corresponds to the complexity of a standard
$(r,\e)$-\pleb query and can be viewed as the incompressible time
needed to locate the query point $q$ in the data structure. The second
term ($\frac{n^\varrho}{p_1}$) bounds the total number of collisions
of $q$ with data points lying outside $\ball(q, r(1+\e))$.  The third
term ($\frac{n^\varrho}{p_1} |\nnpqr{P}{q}{r_0}|$) arises from the
fact that a data point lying within distance $r_0$ of $q$ may
collide in every single hash table with $q$. Finally, the last term
($\frac{n^\alpha}{p_1} |\nnpqr{P}{q}{r(1+\e)}\setminus\ball(q,r_0)|$)
arises from the fact that the points of $\nnpqr{P}{q}{r(1+\e)}$ that
lie farther than $r_0$ can only collide up to $\frac{n^\alpha}{p_1}$
times with $q$ each, for some $\alpha\leq\varrho$. Note that the less
sensitive the family $\f$ between radii $r_0$ and $r_1$, the closer to
$1$ the ratio $\frac{\ln p_0}{\ln p_1}$, and therefore the smaller
$\alpha$ compared to $\varrho$. By contrast, the more sensitive the
family between radii $r_1$ and $r_2$, the smaller the ratio
$\varrho=\frac{\ln p_1}{\ln p_2}$ compared to $1$.

Our proof of Theorem~\ref{thm:allnn} follows previous
literature~\cite{HarPeledIndykMotwani} and is divided into three
parts: (1) proving the correctness of the output of
Algorithm~\ref{alg:nnquery} with high probability, (2) bounding the
expected query time, and (3) bounding the size of the data
structure. The novelty resides in Lemma~\ref{lem:realpos}, which
exploits the axiom (iii) of Definition~\ref{def:stronglshashfunc} to
bound the number of collisions of $q$ with points of
$\nnpqr{P}{q}{r(1+\e)}\setminus\ball(q,r_0)$.

\paragraph{Correctness of the output.}
\label{sec:A_correct}
Note that the test on line~\ref{line:approxnn} of
Algorithm~\ref{alg:nnquery} ensures that the output set $S$ is
always a subset of $\nnpqr{P}{q}{r}$. Thus, we only need to show that
$S$ contains all the points of $\nnpqr{P}{q}{r}$ with high probability
at the end of the query.
\begin{lem} \label{lem:success}
$\nnpqr{P}{q}{r}\subseteq S$ with probability at least
  $1-n^{1-c\ln\frac{5}{2}}$. 
\end{lem}
This result means that the probability of success of the query is
high, even for small values of $c$. For instance, it is at least
$1-\frac{1}{n}$ for $c\geq \frac{2}{\ln\frac{5}{2}}$, and more
generally it is at least $1-\frac{1}{n^\omega}$ for $c\geq
\frac{1+\omega}{\ln\frac{5}{2}}$.
\begin{proof}[Proof of the lemma]
  Let $p$ be a point of $\nnpqr{P}{q}{r}$. Consider a single iteration
  $i$ of the main loop of Algorithm~\ref{alg:nnquery}, and let us show
  that $p$ is inserted in the output set $S$ during this iteration
  with constant probability.  This is equivalent to showing that, with
  constant probability, there exists some function $g_j(\cdot)$ that
  hashes $q$ and $p$ to the same location ($g_j(q)=g_j(p)$). Since
  $d(q,p)\leq r$, the probability of a collision for a fixed $j$ is at
  least $p_1^k = p_1^{\lceil\frac{\ln n}{\ln 1/p_2}\rceil} =
  e^{-\ln{1/p_1}\lceil\frac{\ln n}{\ln 1/p_2}\rceil} \geq
  e^{-\ln{1/p_1}(\frac{\ln n}{\ln 1/p_2}+1)} = p_1 n^{-\frac{\ln
      1/p_1}{\ln 1/p_2}} = p_1 n^{-\varrho}$. Therefore, the
  probability that no hash function $g_j$ generates a collision is at
  most $(1-p_1 n^{-\varrho})^{L}\leq (1-p_1
  n^{-\varrho})^{n^\varrho/p_1}$ since $L=\lceil
  \frac{n^\varrho}{p_1}\rceil$ functions are picked from $\mathcal{G}$
  at iteration $i$. Thus, the probability that this iteration inserts
  $p$ into the output set $S$ is at least
  $1-(1-p_1n^{-\varrho})^{n^\varrho/p_1} \geq 1 - \frac{1}{e} >
  \frac{3}{5}$.

Now, there are $\lceil c\ln n\rceil$ iterations in total, with
independent hash functions, so the probability that $p\notin S$ at the
end of the query is at most $\left(\frac{2}{5}\right)^{\lceil c\ln
  n\rceil}=e^{\ln\frac{2}{5} \lceil c\ln n\rceil} \leq
n^{c\ln\frac{2}{5}}$. Applying the union bound on the set
$\nnpqr{P}{q}{r}$, we obtain that the probability that all points of
$\nnpqr{P}{q}{r}$ belong to $S$ at the end of the query is at least
$1-|\nnpqr{P}{q}{r}| n^{c\ln\frac{2}{5}}\geq 1-n^{1+c\ln\frac{2}{5}} =
1-n^{1-c\ln\frac{5}{2}}$.
\end{proof}

\begin{remark}\label{rem:increase_n}
It is easily seen from the final paragraph of the proof of
Lemma~\ref{lem:success} that the correctness of the output can be
guaranteed with probability $1-m^{1-c\ln\frac{5}{2}}$ for any given
$m\geq n$. Indeed, by running $\lceil c\ln m\rceil$ iterations of the
main loops of Algorithms~\ref{alg:nnpre} and~\ref{alg:nnquery} instead
of $\lceil c \ln n\rceil$ iterations, we obtain that each point of
$\nnpqr{P}{q}{r}$ belongs to $S$ at the end of the query with probability
at least $1-m^{-c\ln\frac{5}{2}}$, and thus that $\nnpqr{P}{q}{r}\subseteq
S$ with probability at least $1-m^{1-c\ln \frac{5}{2}}$. This
remark will be useful when dealing with \rnn queries in
Section~\ref{sec:rnn}.
\end{remark}

\paragraph{Expected query time.}
First of all, the query point $q$ is hashed into $\lceil
\frac{n^\varrho}{p_1}\rceil \lceil c\ln n\rceil=\tilde
O(\frac{n^\varrho}{p_1})$ hash tables in total, and each hashing
operation involves $k=\lceil \frac{\ln
  n}{\ln\nicefrac{1}{p_2}}\rceil=O(\frac{\ln n}{\ln
  \nicefrac{1}{p_2}})$ hash function evaluations, $c$ being a constant
here. Thus, the total number of hash function evaluations is $\tilde
O(\frac{n^\varrho}{p_1 \ln\nicefrac{1}{p_2}})$, and so is the total time
spent hashing $q$ (modulo the time needed to do a hash function
evaluation, which is ignored here as in the previous sections). There
remains to bound the expected number of colllisions of $q$ with points
of $P$ in the hash tables.
\begin{lem} \label{lem:falsepos}
The expected total number of collisions of $q$ with points of
$P\setminus \ball(q,r(1+\e))$ is $\tilde O(\frac{n^\varrho}{p_1})$.
\end{lem}
\begin{proof}
 Take an arbitrary iteration $i$ of the main loop of
 Algorithm~\ref{alg:nnquery}, and an arbitrary hash table $H_j$
 considered during that iteration. Recall that the hash family $\g$ is
 constructed in Algorithm~\ref{alg:nnpre} by concatenating
 $k=\lceil\frac{\ln n}{\ln\nicefrac{1}{p_2}}\rceil$ functions drawn
 from a $(r,(1+\e)r,p_1,p_2)$-sensitive family $\f$. Therefore, the
 probability that a given point of $P\setminus \ball(q,r(1+\e))$
 collides with $q$ in $H_j$ is at most $p_2^k = p_2^{\lceil\frac{\ln
     n}{\ln 1/p_2}\rceil} = e^{-\ln 1/p_2 \lceil\frac{\ln n}{\ln
     1/p_2}\rceil} \leq e^{-\ln 1/p_2 \frac{\ln n}{\ln 1/p_2}} =
 \frac{1}{n}$.  It follows that the expected number of points
 of $P\setminus \ball(q,r(1+\e))$ that collide with $q$ in $H_j$ is
 at most $1$, from which we conclude that the expected
 total number of such collisions in all the hash tables at all
 iterations is at most $\lceil \frac{n^\varrho}{p_1}\rceil \lceil
 c\ln n\rceil = \tilde O(\frac{n^\varrho}{p_1})$.
\end{proof}

Without any further assumptions on the family $\f$ of hash functions,
each point of $\nnpqr{P}{q}{r(1+\e)}$ might collide with $q$ in every
hash table. The number of collisions of $q$ with points of
$\nnpqr{P}{q}{r(1+\e)}$ is therefore $O(\frac{n^\varrho}{p_1} c\ln n
|\nnpqr{P}{q}{r(1+\e)}|) = \tilde
O(\frac{n^\varrho}{p_1}|\nnpqr{P}{q}{r(1+\e)}|)$. Combined with
Lemma~\ref{lem:falsepos}, this bound implies that the expected running
time of the algorithm is $\tilde O(\frac{n^\varrho}{p_1}
(\frac{1}{\ln\nicefrac{1}{p_2}}+1+|\nnpqr{P}{q}{r(1+\e)}|))$,
as claimed in the theorem.
For every collision considered, a test is made on the distance between
$q$ and the colliding point of $P$ (see line~\ref{line:approxnn} of
Algorithm~\ref{alg:nnquery}). With a simple book-keeping, e.g. by
marking the points of $P$ that have already been considered during the
query, we can afford to do the test at most once per point of $P$,
thus yielding a total number of distance computations of the order of
$\tilde O(\frac{n^\varrho}{p_1} + |\nnpqr{P}{q}{r(1+\e)}|)$.

Consider now the stronger hypothesis that the family $\f$ of hash
functions is $(r_0, r, r(1+\e), p_0, p_1, p_2)$-sensitive for some
$r_0\leq r$.
\begin{lem} \label{lem:realpos}
Assuming that $\f$ is $(r_0, r, r(1+\e), p_0, p_1, p_2)$-sensitive,
the expected total number of collisions of $q$ with points of
$\nnpqr{P}{q}{r(1+\e)}\setminus \ball(q,r_0)$ is $\tilde
O(\frac{n^\alpha}{p_1} |\nnpqr{P}{q}{r(1+\e)}\setminus
\ball(q,r_0)|)$, where $\alpha=\varrho(1-\frac{\ln p_0}{\ln
  p_1})\leq\varrho$.
\end{lem}
\begin{proof}
 Take an arbitrary iteration $i$ of the main loop of
 Algorithm~\ref{alg:nnquery}, and an arbitrary hash table $H_j$
 considered during that iteration. The probability that a given point
 $p\in \nnpqr{P}{q}{r(1+\e)}\setminus \ball(q,r_0)$ collides with $q$
 in $H_j$ is at most $p_0^k = p_0^{\lceil\frac{\ln n}{\ln
     1/p_2}\rceil} = e^{\ln p_0\lceil\frac{\ln n}{\ln 1/p_2}\rceil}
 \leq e^{\ln p_0 \frac{\ln n}{\ln 1/p_2}} = 
 n^{-\frac{\ln p_0}{\ln p_2}}$.  It follows that the expected total
 number of collisions between $p$ and $q$ during the execution of the
 algorithm is at most $n^{-\frac{\ln p_0}{\ln p_2}}
 \lceil \frac{n^\varrho}{p_1}\rceil \lceil c\ln n\rceil = \tilde
 O(\frac{n^\alpha}{p_1})$, where $\alpha=\varrho-\frac{\ln p_0}{\ln
   p_2} = \frac{\ln p_1}{\ln p_2}-\frac{\ln p_0}{\ln p_2} = \frac{\ln
   p_1}{\ln p_2} (1 - \frac{\ln p_0}{\ln p_1})$. We conclude that the
 expected total number of collisions of $q$ with points of
 $\nnpqr{P}{q}{r(1+\e)}\setminus \ball(q,r_0)$ during the course of
 the algorithm is $\tilde O(\frac{n^\alpha}{p_1}
 |\nnpqr{P}{q}{r(1+\e)}\setminus \ball(q,r_0)|)$.
\end{proof}

It follows from Lemma~\ref{lem:realpos} that the expected query time
becomes $\tilde O(\frac{n^\varrho}{p_1}
(\frac{1}{\ln\nicefrac{1}{p_2}}+1+|\nnpqr{P}{q}{r_0}|) +
\frac{n^\alpha}{p_1}|\nnpqr{P}{q}{r(1+\e)}\setminus \ball(q,r_0)|)$ when the
family $\f$ of hash functions is $(r_0, r, r(1+\e), p_0, p_1,
p_2)$-sensitive, as claimed in the theorem.

\paragraph{Size of the data structure.}
Each hash table contains one pointer per point of $P$, and there are
$\lceil \frac{n^\varrho}{p_1}\rceil \lceil c\ln n\rceil$ such hash
tables in total, so we need to store $\tilde
O(\frac{n^{1+\varrho}}{p_1})$ pointers in total. In addition, we need
to store the $\lceil \frac{n^\varrho}{p_1}\rceil \lceil c\ln n\rceil$
vectors of hash functions corresponding to the hash tables, but this
term is dominated by the previous one. Thus, in total our data
structure has a space complexity of $\tilde
O(\frac{n^{1+\varrho}}{p_1})$. This bound ignores the costs of storing
the input point cloud and the selected hash functions, which depend on
the type of data representation.

\subsection{Affine case: the non-isometric embedding trick}
\label{sec:erpleb_Rd}

Assume from now on that the ambient space is $(\R^d,\ell_s)$, where
$s\in (0,2]$, and note that axiom~(iii) of
Definition~\ref{def:stronglshashfunc} is satisfied by the families of
hash functions introduced in Section~\ref{sec:pleb_affine} since the
probability $\Prob(l)$ defined in Eq.~(\ref{eq:proba_collision_ls})
decreases as the distance $l$ increases. In order to prevent the
points of $P$ from getting too close to the query point $q$, so
axiom~(iii) can be exploited, our strategy is to 
apply a non-isometric embedding into $(\R^{d+1}, \ell_s)$ that moves
$q$ away from $P$, while preserving the order of the distances to $q$.

At preprocessing time, we lift the points of $P$ to $(\R^{d+1},
\ell_s)$ by adding one coordinate equal to $0$ to every point. We then
build an $\A(P', r', \e')$ data structure using
Algorithm~\ref{alg:nnpre}, where $P'$ denotes the image of $P$ through
the embedding, $r'=r(1+\frac{1}{(1+\e)^s-1})^{1/s}$, and
$\e'=((1+\e)^s+(1+\e)^{-s}-1)^{1/s}-1$.  In effect, right before
building the data structure we follow
Section~\ref{sec:pleb_affine} and rescale $P'$ by a factor of $1/r'$,
to get a normalized point cloud $P''$ on top of which we build
an $\A(P'', 1, \e')$ data structure using Algorithm~\ref{alg:nnpre}.

At query time, we lift $q$ to $\R^{d+1}$ by adding one coordinate
equal to $\frac{r}{((1+\e)^s-1)^{1/s}}$, then we answer an exhaustive
$r'$-\pleb query in $\R^{d+1}$ by running Algorithm~\ref{alg:nnquery}
with the $\A(P', r', \e')$ data structure, and then we return the
pre-image of the output set through the embedding. Once again, in
effect we rescale the image of the query point in $\R^{d+1}$ by a
factor of $1/r'$, so Algorithm~\ref{alg:nnquery} is actually run
with $\A(P'', 1, \e')$.

Note that the embedding into $\R^{d+1}$ is not isometric since it
does not preserve the distances of $q$ to the data points. However, it
does preserve their order. Indeed, for every point $p\in P$ the distance
$\dist(p,q)$ becomes $(\dist(p,q)^s + \frac{r^s}{(1+\e)^s-1})^{1/s}$
after the embedding. Since the map
$t\mapsto (t^s + \frac{r^s}{(1+\e)^s-1})^{1/s}$ is monotonically
increasing with $t$, the embedding preserves the order of
distances to $q$. We then have the following easy properties, where $x'\in
\R^{d+1}$ denotes the image of any point $x\in P\cup\{q\}$ through the
embedding:
\begin{itemize}
\item[\rm (i)] $\forall p\in P$, $\dist(p,q)\leq r \Leftrightarrow
  \dist(p', q')\leq
  \left(r^s+\frac{r^s}{(1+\e)^s-1}\right)^{1/s}=r\left(1+\frac{1}{(1+\e)^s-1}\right)^{1/s} = r'$;
\item[\rm (ii)] $\forall p\in P$, $\dist(p', q')\geq
  \left(\frac{r^s}{(1+\e)^s-1}\right)^{1/s} =
  \frac{r}{1+\e}\left(\frac{(1+\e)^s}{(1+\e)^s-1}\right)^{1/s} = \frac{r}{1+\e}
  \left(1+\frac{1}{(1+\e)^s-1}\right)^{1/s} = \frac{r'}{1+\e}$;
\item[\rm (iii)] $\forall p\in P$, $\dist(p', q')\leq r'(1+\e')
  \Rightarrow \dist(p, q)\leq
  \left(r'^s(1+\e')^s-\frac{r^s}{(1+\e)^s-1}\right)^{1/s} =$\\$
  \left(r^s(1+\frac{1}{(1+\e)^s-1})((1+\e)^s + (1+\e)^{-s}
  -1)-\frac{r^s}{(1+\e)^s-1}\right)^{1/s} =$\\$ 
  r\left(\frac{(1+\e)^s}{(1+\e)^s-1}((1+\e)^s+(1+\e)^{-s}-1) - \frac{1}{(1+\e)^s-1}\right)^{1/s} =$\\$ r\left(1+\e\right)$.
\end{itemize}
It follows from (i) that $\nnpqr{P'}{q'}{r'}$ is the image of
$\nnpqr{P}{q}{r}$ through the embedding. Hence, by
Lemma~\ref{lem:success}, with high probability the output set of the
exhaustive $r'$-\pleb query in $\R^{d+1}$ is the image of
$\nnpqr{P}{q}{r}$ through the embedding. Thus, our output is correct
with high probability. In the meantime, the embedding has the
following impact on the complexity bounds of Theorem~\ref{thm:allnn}:
\begin{itemize}
\item On the negative side, parameter $\e$ is now replaced by
  $\e'=((1+\e)^s + (1+\e)^{-s} - 1)^{1/s} - 1 \leq \e$, which
  increases $p_2$ from $\Prob(1+\e)$ to $\Prob(1+\e')$. This means
  that the ratio $\varrho = \frac{\ln p_1}{\ln p_2}$ becomes
  $\frac{\ln \Prob(1)}{\ln \Prob(1+\e')} \geq \frac{\ln \Prob(1)}{\ln
    \Prob(1+\e)}$ and thus gets closer to $1$, even though it still
  remains strictly below $1$. Furthermore, the term $\frac{1}{\ln
    \nicefrac{1}{p_2}}$ grows from
  $\frac{1}{\ln \nicefrac{1}{\Prob(1+\e)}}$ to $\frac{1}{\ln
    \nicefrac{1}{\Prob(1+\e')}}$.
\item On the positive side, we know from (ii) that the points of $P'$
  lie at least $\frac{r'}{1+\e}$ away from the query point $q$, so by
  Lemma~\ref{lem:realpos} they cannot collide with $q$ more than
  $\tilde O(\frac{n^\alpha}{p_1})$ times each in expectation, where
  $\alpha=\varrho(1-\frac{\ln p_0}{\ln p_1})$, $p_1=\Prob(1)$, and
  $p_0=\Prob(\frac{1}{1+\e})$.
\end{itemize}
For the rest, the embedding is a neutral operation. Indeed, even
though the complexity now depends on the size of
$\nnpqr{P'}{q'}{r'(1+\e')}$ instead of the size of
$\nnpqr{P}{q}{r(1+\e)}$, we know from (iii) that the preimage of the
former set through the embedding is contained within the latter set,
so we have $|\nnpqr{P'}{q'}{r'(1+\e')}|\leq |\nnpqr{P}{q}{r(1+\e)}|$.
In addition, the fact that the query now takes place in $\R^{d+1}$
instead of $\R^d$, with a radius parameter that grew from $r$ to $r'$,
does not affect the probabilities $p_1, p_2$, which depend neither on
the ambient dimension as pointed out after
Eq.~(\ref{eq:proba_collision_ls}), nor on the radius thanks to the
rescaling of the data. It also does not affect the asymptotic complexities
of distance computations and hash function evaluations, which remain
$O(d)$.

All in all, we obtain the following complexity bounds for the
exhaustive $r$-\pleb query in $(\R^d, \ell_s)$, where $\A'(P,r,\e)$
denotes the full data structure built at preprocessing time, which
contains the embedding and rescaling information together with the
$\A(P'', 1, \e')$ data structure:
\begin{thm} \label{thm:allnn_Rd}
  Given a finite set $P$ with $n$ points in $(\R^d, \ell_s)$, $s\in
  (0,2]$, and two parameters $r,\e> 0$, the $\A'(P, r, \e)$ data
    structure answers exhaustive $r$-\pleb queries correctly with high
    probability in expected $\tilde O(\frac{n^\varrho}{p_1}(\frac{1}{\ln
      \nicefrac{1}{p_2}}+1)+
    \frac{n^\alpha}{p_1}|\nnpqr{P}{q}{r(1+\e)}|)$ time using $\tilde
    O(\frac{n^{1+\varrho}}{p_1})$ space, where $\varrho=\frac{\ln p_1}{\ln p_2}$
    and $\alpha=\varrho(1-\frac{\ln p_0}{\ln p_1})\leq\varrho$, the
    quantities $p_0=\Prob(\frac{1}{1+\e})$, $p_1=\Prob(1)$ and
    $p_2=\Prob(((1+\e)^s + (1+\e)^{-s} - 1)^{1/s})$ being derived from
    some $s$-stable distribution $D$ according to
    Eq.~(\ref{eq:proba_collision_ls}).
\end{thm}
Quantifying precisely the amounts by which the quantities $\varrho$,
$\alpha$ and $\frac{1}{\ln\nicefrac{1}{p_2}}$ are affected by the
embedding, what the corresponding {\em best} choice of parameter $w$
is, and how this choice impacts $\frac{1}{p_1}$, are the main
questions at this point. Because Eq~(\ref{eq:proba_collision_ls}) may
not always have a closed form solution, it is difficult to provide an
answer in full generality for all values $s\in (0,2]$. We will
  nevertheless investigate two special cases that are of practical
  interest: $s=1$ and $s=2$.

\begin{figure}[!htb]
\centering
\includegraphics[height=0.48\textwidth]{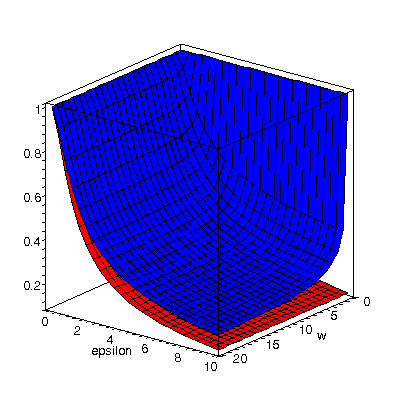}
\includegraphics[height=0.48\textwidth]{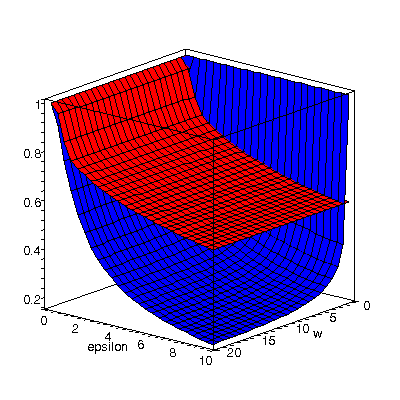}
\caption{\em Behavior of $\varrho$ in $(\R^{d+1}, \ell_1)$. Left: plots
  of $\varrho$ (blue) and $\frac{1}{1+\e'}=\frac{1}{1+\e^2/(1+\e)}$
  (red) versus $\e$ and $w$.  Right: plots of $\varrho$ (blue) and
  $\frac{1}{1+\min\{\e^2,\;\sqrt{\e}\}/4}$ (red) versus $\e$ and $w$.}
\label{fig:l1_rho_theorique}
\end{figure}

\begin{figure}[!htb]
\centering
\includegraphics[height=0.48\textwidth]{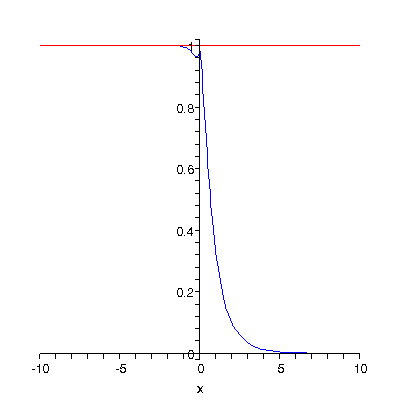}
\includegraphics[height=0.48\textwidth]{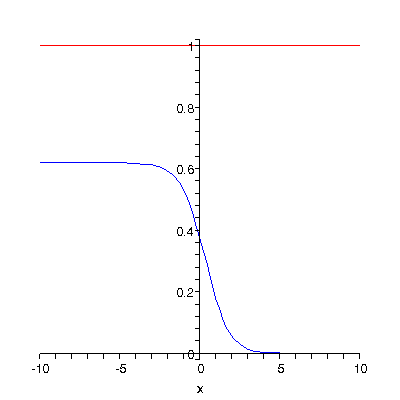}
\caption{\em Behaviors of $\varrho$ and $\alpha$ in $(\R^{d+1},
  \ell_1)$ after letting $w=\max\{1,\e\}$. From left to right, in
  blue: plots of $\varrho (1+\min\{\e^2,\;\sqrt{\e}\}/4)$ and
  $\frac{\alpha}{\e\varrho}$. Both plots are versus
  $\e$ on a logarithmic scale ($x=\log_{10} \e$). The red lines have
  equation $y=1$.}
\label{fig:l1_rho_alpha}
\end{figure}

\begin{figure}[!htb]
\centering
\includegraphics[height=0.48\textwidth]{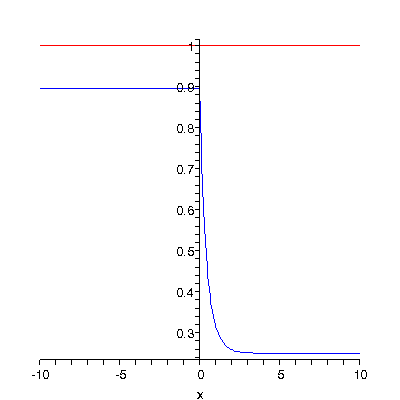}
\includegraphics[height=0.48\textwidth]{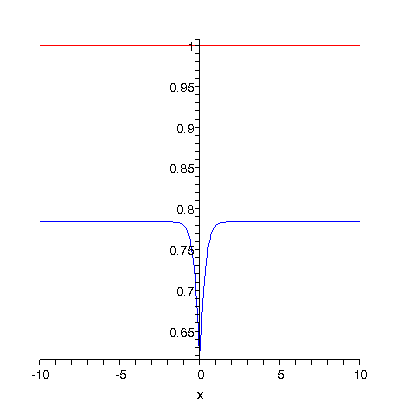}
\caption{\em Behaviors of $\frac{1}{\Prob(1)}$ and
  $\frac{1}{\ln\nicefrac{1}{\Prob(1+\e')}}$ in $(\R^{d+1}, \ell_1)$ after
  letting $w=\max\{1,\e\}$. From left to right, in blue: plots of
  $\frac{1}{4\Prob(1)}$ and $\frac{1}{\ln\nicefrac{1}{\Prob(1+\e')}}$. Both
  plots are versus $\e$ on a logarithmic scale ($x=\log_{10} \e$). The
  red lines have equation $y=1$.}
\label{fig:l1_p2_k}
\end{figure}

\paragraph{Case $s=1$.}
The definition of $\e'$ gives $\e'=\frac{\e^2}{1+\e}$ in this
case. The formula for $\varrho$ is then the same as in $\R^d$, with
$\e$ replaced by $\frac{\e^2}{1+\e}$. As reported in~\cite{DIIM04} and
illustrated in Figure~\ref{fig:l1_rho_theorique} (left), $\varrho$
remains above $\frac{1}{1+\e^2/(1+\e)}$, even though it seems to
converge to this quantity as $w$ tends to infinity. Letting
$w=\max\{1,\e\}$, we found experimentally that $\varrho$ is dominated
by $\frac{1}{1+\e^2/4}$ when $\e\leq 1$ and by
$\frac{1}{1+\sqrt{\e}/4}$ when $\e\geq 1$, as shown in
Figures~\ref{fig:l1_rho_theorique} (right) and~\ref{fig:l1_rho_alpha}
(left). In the meantime, $\alpha$ is less than $\e\varrho$, as
can be seen from Figure~\ref{fig:l1_rho_alpha} (right), while
$\frac{1}{\Prob(1)}$ and $\frac{1}{\ln\nicefrac{1}{\Prob(1+\e')}}$ are less
than $4$ and $1$ respectively, as shown in Figure~\ref{fig:l1_p2_k}.
All in all, Theorem~\ref{thm:allnn_Rd} can be re-written as follows:
\addtocounter{thm}{-1}
\begin{thm}[case $s=1$] \label{thm:allnn_Rd_l1}
  Given a finite set $P$ with $n$ points in $(\R^d, \ell_1)$, and two
  parameters $r,\e> 0$, the $\A'(P, r, \e)$ data
    structure answers exhaustive $r$-\pleb
  queries correctly with high probability in expected $\tilde
  O(n^\varrho+ n^\alpha|\nnpqr{P}{q}{r(1+\e)}|)$ time using
    $\tilde O(n^{1+\varrho})$ space, where $\varrho\leq
    \frac{1}{1+\min\{\e^2,\;\sqrt{\e}\}/4}<1$ and $\alpha\leq
    \e\varrho < \e$.
\end{thm}

\begin{figure}[!htb]
\centering
\includegraphics[height=0.48\textwidth]{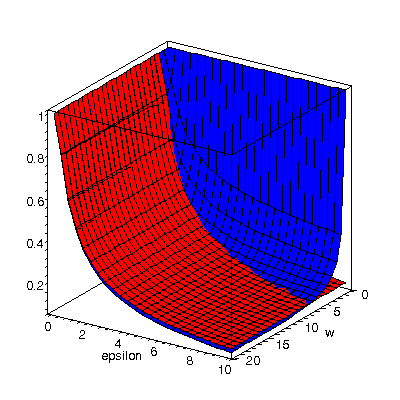}
\includegraphics[height=0.48\textwidth]{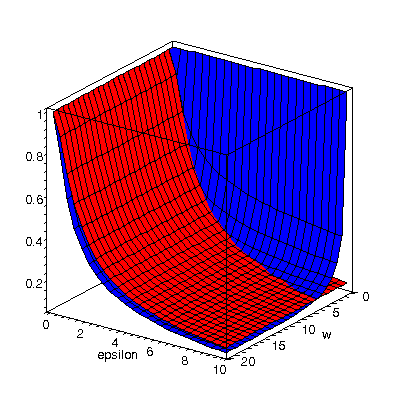}
\caption{\em Behavior of $\varrho$ in $(\R^{d+1}, \ell_2)$. Left: plots
  of $\varrho$ (blue) and
  $\frac{1}{1+\e'}=\frac{1+\e}{\sqrt{(1+\e)^4-1}}$ (red) versus $\e$
  and $w$.  Right: plots of $\varrho$ (blue) and
  $\frac{1}{1+\e^2/(1+\e)}$ (red) versus $\e$ and $w$.}
\label{fig:l2_rho_theorique}
\end{figure}

\begin{figure}[!htb]
\centering
\includegraphics[height=0.48\textwidth]{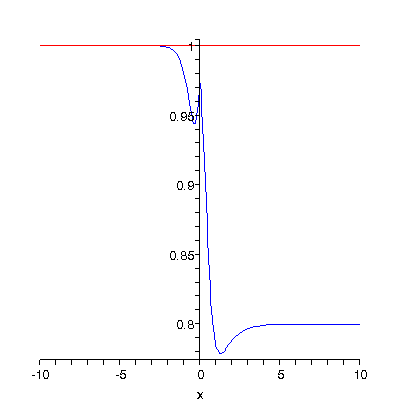}
\includegraphics[height=0.48\textwidth]{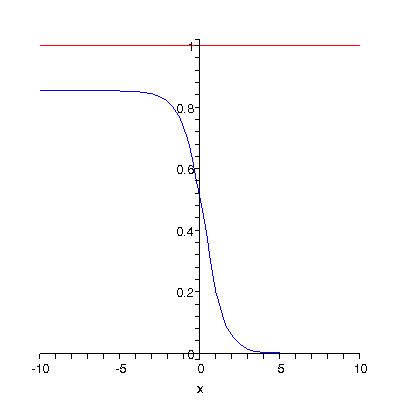}
\caption{\em Behaviors of $\varrho$ and $\alpha$ in $(\R^{d+1},
  \ell_2)$ after letting $w=\max\{1,\e\}$. From left to right, in
  blue: plots of $\varrho (1+\e^2/(1+\e))$ and
  $\frac{\alpha}{\e\varrho}$. Both plots are versus
  $\e$ on a logarithmic scale ($x=\log_{10} \e$). The red lines have
  equation $y=1$.}
\label{fig:l2_rho_alpha}
\end{figure}

\begin{figure}[!htb]
\centering
\includegraphics[height=0.48\textwidth]{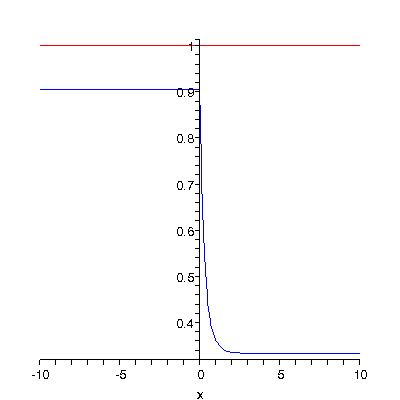}
\includegraphics[height=0.48\textwidth]{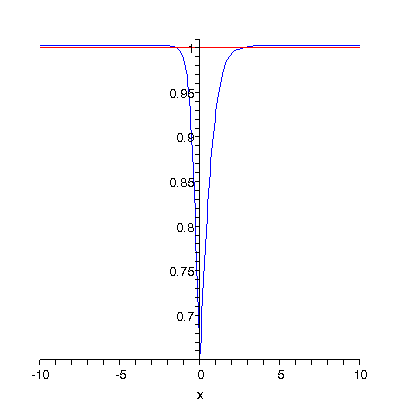}
\caption{\em Behaviors of $\frac{1}{\Prob(1)}$ and
  $\frac{1}{\ln\nicefrac{1}{\Prob(1+\e')}}$ in $(\R^{d+1}, \ell_2)$ after
  letting $w=\max\{1,\e\}$. From left to right, in blue: plots of
  $\frac{1}{3\Prob(1)}$ and
  $\frac{1}{\ln\nicefrac{1}{\Prob(1+\e')}}$. Both plots are versus
  $\e$ on a logarithmic scale ($x=\log_{10} \e$). The red lines have
  equation $y=1$.}
\label{fig:l2_p2_k}
\end{figure}

\paragraph{Case $s=2$.}
The definition of $\e'$ gives $\e'=\frac{\sqrt{(1+\e)^4-1}}{1+\e}-1$
in this case. The formula for $\varrho$ is then the same as in $\R^d$,
with $\e$ replaced by $\frac{\sqrt{(1+\e)^4-1}}{1+\e}-1$. As pointed
out in~\cite{DIIM04} and illustrated in
Figure~\ref{fig:l2_rho_theorique} (left), $\varrho$ goes below
$\frac{1+\e}{\sqrt{(1+\e)^4-1}}$ at reasonably small values of
parameter $w$. Since this bound is not quite evocative, we used a
slightly different bound, namely $\frac{1}{1+\e^2/(1+\e)}$, and we
found experimentally that $\varrho\leq \frac{1}{1+\e^2/(1+\e)}$
whenever $w=\max\{1,\e\}$, as shown in
Figures~\ref{fig:l2_rho_theorique} (right) and~\ref{fig:l2_rho_alpha}
(left). In the meantime, $\alpha$ is less than $\e\varrho$, as can be
seen from Figure~\ref{fig:l2_rho_alpha} (right), while the terms
$\frac{1}{\Prob(1)}$ and $\frac{1}{\ln\nicefrac{1}{\Prob(1+\e')}}$ are bounded
by small constants, as shown in Figure~\ref{fig:l2_p2_k}. All in all,
Theorem~\ref{thm:allnn_Rd} can be re-written as follows:
\addtocounter{thm}{-1}
\begin{thm}[case $s=2$] \label{thm:allnn_Rd_l2}
  Given a finite set $P$ with $n$ points in $(\R^d, \ell_2)$, and two
  parameters $r,\e> 0$, the $\A'(P, r, \e)$ data
    structure answers exhaustive $r$-\pleb
  queries correctly with high probability in expected $\tilde
  O(n^\varrho+ n^\alpha|\nnpqr{P}{q}{r(1+\e)}|)$ time using
    $\tilde O(n^{1+\varrho})$ space, where $\varrho\leq
    \frac{1}{1+\e^2/(1+\e)}<1$ and $\alpha\leq
    \e\varrho < \e$.
\end{thm}

\section{Interlude: from exhaustive $r$-\pleb to exact \nn}
\label{sec:exact-nn}

Before dealing with \rnn queries (the main topic of the paper),
let us show a simple but pedagogical application of exhaustive
$r$-\pleb queries to exact $\nn$ search. Given a set $P$ with $n$ points
and a user-defined parameter $\e>0$, we will show that \nn queries can
be solved exactly with high probability on any query point $q$ in
expected $\tilde O(n^\varrho + n^\alpha |\ennpq{O(\e)}{P}{q}|)$ time
using $\tilde O(n^{1+\varrho})$ space, for some quantities
$\varrho= \frac{1}{1+\Theta(\e^2)}<1$ and $\alpha\leq\e\varrho<\e$
(Theorem~\ref{thm:exact-nn}). The running time bound is composed of
two terms: the first one is sublinear in $n$ and corresponds to a
standard approximate $\e$-\nn query using locality-sensitive hashing;
the second one depends on the size of the approximate nearest
neighbors set $\ennpq{O(\e)}{P}{q}$ and indicates that the solution to
the exact query is sought for among this set. Whether the bound will
be sublinear in $n$ or not in the end depends on the size of the set
compared to the quantity $n^{1-\alpha}$. This follows the intuition
that finding the exact nearest neighbor of $q$ is easy when $q$ does
not have too many approximate nearest neighbors, and in this respect
the quantity $|\ennpq{O(\e)}{P}{q}|$ plays the role of a {\em
  condition number} measuring the inherent difficulty of a given
instance of the exact $\nn$ problem. The interesting point to raise
here is that the limit on this number for our algorithm to be
sublinear is at least of the order of $n^{1-\e}$ since we
have $\alpha<\e$.

Let us point out that the above bounds are for the ambient space
$\R^d$ equipped with the $\ell_1$- or $\ell_2$-norm. Our analysis will
be carried out in the more general setting of an $\ell_s$-norm, with
$s\in (0,2]$, where we will derive more general complexity bounds.
  The choice of $(\R^d, \ell_s)$ is mainly for ease of
  exposition, since the algorithm can actually be applied in arbitrary
  metric spaces that admit locality-sensitive families of hash
  functions, where its analysis extends in a straightforward manner
  (see Remark~\ref{rem:exact-NN_general} at the end of the section).

\paragraph{The algorithm.} Let $P$ be a finite set of $n$ points in $(\R^d, \ell_s)$, $s\in (0,2]$, and let $\e>0$ be a parameter. The preprocessing phase
consists of the following steps:
\begin{slist}
\item[i.] Build the tree structure $\T(P,\e)$ of
  Section~\ref{sec:renn-to-enn} and its associated $(r,\e)$-\pleb data
  structures.
\item[ii.] For every $(r,\e)$-\pleb data structure built on some subset
  of $P$ at step i, build an $\A'(P, r,\e)$ data structure using the
  procedure of Section~\ref{sec:erpleb_Rd}.
\end{slist}
Then, given a query point $q$, we proceed as follows:
\begin{slist}
\item[1.] Answer an $\e$-\nn query using the tree structure $\T(P,\e)$,
  and let $r\geq 0$ be the output value.
\item[2.] Answer an exhaustive $r$-\pleb query using the $\A'(P,r,\e)$
  data structure, and let $S$ be the output set.
\item[3.] Iterate over the points of $S$ and return the one that is
  closest to $q$. If $S$ is empty, then return any arbitrary point of $P$.
\end{slist}
Note that the execution of step 2 is
made possible by the fact that the algorithm solving the $\e$-\nn
query at step 1 returns a radius $r$ that is stored in one of the
$\A'(P,r,\e)$ data structures built during the preprocessing
phase. For any other value $r$ we would not be able to perform step 2
because we would not have the corresponding $\A'(P,r,\e)$ data
structure at hand.

\paragraph{Analysis.} We begin by showing the correctness of the 
query procedure:

\begin{lem}\label{lem:exact-nn_correct}
The query procedure returns a point of $\nnp{q}$ with high probability.
\end{lem}
\begin{proof}
Corollary~\ref{cor:lsh4enn} guarantees that the radius $r$ computed at
step 1 satisfies $\dist(q,P)\leq r \leq \dist(q,P)(1+\e)$ with high
probability. Under this condition, we have $\nnp{q}\subseteq
\nnpqr{P}{q}{r}$, and so Theorem~\ref{thm:allnn_Rd} guarantees that
the set $S$ computed at step 2 contains $\nnp{q}$ with high
probability. It follows that the point returned at step 3 belongs to
$\nnp{q}$ with high probability.
\end{proof}

We will now analyze the expected running time of the query. Let $D$ be
the $s$-stable distribution used by the algorithm, and let
$p_0=\Prob(\frac{1}{1+\e})$, $p_1=\Prob(1)$, $p_2=\Prob(1+\e)$ and
$p'_2=\Prob(((1+\e)^s+(1+\e)^{-s}-1)^{1/s})$ be derived from $D$ according
to Eq.~(\ref{eq:proba_collision_ls}). By Corollary~\ref{cor:lsh4enn},
the running time of step 1 is $\tilde O(\frac{n^\varrho}{p_1 \ln
  \nicefrac{1}{p_2}})$, where $\varrho=\frac{\ln p_1}{\ln p_2}$. The
running time of step 3 is $O(|S|)$, so it is dominated by the running
time of step 2.
\begin{lem}\label{lem:exact-nn_runtime}
The expected running time of step 2 is $\tilde
O(\frac{n^{\varrho'}}{p_1}(\frac{1}{\ln\nicefrac{1}{p'_2}} + 1) +
\frac{n^\alpha}{p_1}|\ennpq{\e(2+\e)}{P}{q}|)$, where
$\varrho'=\frac{\ln p_1}{\ln p'_2}$ and $\alpha=\varrho' (1-\frac{\ln
  p_0}{\ln p_1})$.
\end{lem}
\begin{proof}
Let $r$ be the radius computed at step 1. By
Theorem~\ref{thm:allnn_Rd}, the expected running time of step 2 is
$\tilde O(\frac{n^{\varrho'}}{p_1}(\frac{1}{\ln\nicefrac{1}{p'_2}}+1)
+ \frac{n^\alpha}{p_1} |\nnpqr{P}{q}{r(1+\e)}|)$. If $r\leq \dist(q,P)(1+\e)$,
then we have $\nnpqr{P}{q}{r(1+\e)}\subseteq \ennpq{\e(2+\e)}{P}{q}$
and so the expected running time becomes $\tilde
O(\frac{n^{\varrho'}}{p_1}(\frac{1}{\ln\nicefrac{1}{p'_2}}+1) +
\frac{n^\alpha}{p_1} |\ennpq{\e(2+\e)}{P}{q}|)$.  By contrast, if $r>
\dist(q,P)(1+\e)$, then we have no bound on the size of
$\nnpr{q}{r(1+\e)}$ other than $n$, so the expected running time of
step 2 becomes $\tilde
O(\frac{n^{\varrho'}}{p_1}(\frac{1}{\ln\nicefrac{1}{p'_2}}+1) +
\frac{n^{\alpha+1}}{p_1})$. Now, recall from Section~\ref{sec:prelim} that the
event that $r> \dist(q,P)(1+\e)$ only occurs with very low
probability, more precisely with probability at most
$\frac{1}{n}$. Therefore, in total the expected running time of step 2
is bounded by $\tilde
O(\frac{n^{\varrho'}}{p_1}(\frac{1}{\ln\nicefrac{1}{p'_2}}+1) +
\frac{n^\alpha}{p_1} |\ennpq{\e(2+\e)}{P}{q}| + \frac{1}{n}
\frac{n^{\alpha+1}}{p_1})$, which is $\tilde
O(\frac{n^{\varrho'}}{p_1}(\frac{1}{\ln\nicefrac{1}{p_2}}+ 1) +
\frac{n^\alpha}{p_1} |\ennpq{\e(2+\e)}{P}{q}|)$ since the set
$\ennpq{\e(2+\e)}{P}{q}$ contains at least one point, namely the
nearest neighbor of $q$.
\end{proof}

Let us now focus on the size of the data structure. By
Corollary~\ref{cor:lsh4enn}, the total size of the tree $\T(P,\e)$ and
associated $(r,\e)$-\pleb data structures is $\tilde O(\frac{1}{\e}
\frac{n^{1+\varrho}}{p_1})$. In addition, since $\T(P,\e)$ has $\tilde O(n)$ nodes
in total, each one storing $\tilde O(\frac{1}{\e})$ data structures
for $(r,\e)$-\pleb, the total number of $\A'(P,r,\e)$ data structures
built at step ii of the preprocessing phase is $\tilde
O(\frac{n}{\e})$. Therefore, by Theorem~\ref{thm:allnn_Rd}, the total
memory usage of the $\A'(P,r,\e)$ data structures is $\tilde
O(\frac{1}{\e}\frac{n^{2+\varrho'}}{p_1})$. 

Observing now that we have $p'_2\geq p_2$ and $\varrho'\geq \varrho$
since $((1+\e)^s+(1+\e)^{-s}-1)^{1/s}\leq 1+\e$, we conclude that our
procedure has the following space and time complexities (where $p'_2$
and $\varrho'$ have been renamed respectively $p_2$ and $\varrho$ for
convenience):
\begin{thm}\label{thm:exact-nn}
  Given a finite set $P$ with $n$ points in $(\R^d,\ell_s)$, $s\in
  (0,2]$, and a user-defined parameter $\e> 0$, our procedure
    answers exact \nn queries with high probability in expected
    $\tilde O(\frac{n^\varrho}{p_1}(\frac{1}{\ln\nicefrac{1}{p_2}}+1)
    + \frac{n^\alpha}{p_1} |\ennpq{\e(2+\e)}{P}{q}|)$ time using $\tilde
    O(\frac{1}{\e}\frac{n^{2+\varrho}}{p_1})$ space, where $\varrho=\frac{\ln
      p_1}{\ln p_2}$ and $\alpha=\varrho(1-\frac{\ln p_0}{\ln p_1})$,
    the quantities $p_0=\Prob(\frac{1}{1+\e})$, $p_1=\Prob(1)$ and
    $p_2=\Prob(((1+\e)^s+(1+\e)^{-s}-1)^{1/s})$ being derived from some
    $s$-stable distribution $D$ according to
    Eq.~(\ref{eq:proba_collision_ls}).
\end{thm}
Replacing Theorem~\ref{thm:allnn_Rd} by its specialized versions for
$s=1$ and $s=2$ in the analysis immediately gives the following complexity
bounds:
\addtocounter{thm}{-1}
\begin{thm}[case $s=1$]\label{thm:exact-nn_l1}
  Given a finite set $P$ with $n$ points in $(\R^d,\ell_1)$, and a
  user-defined parameter $\e> 0$, our procedure answers exact \nn
  queries with high probability in expected $\tilde
  O(n^\varrho + n^\alpha
  |\ennpq{\e(2+\e)}{P}{q}|)$ time using $\tilde
  O(\frac{1}{\e}n^{2+\varrho})$ space, where $\varrho\leq
  \frac{1}{1+\min\{\e^2,\;\sqrt{\e}\}/4}<1$ and $\alpha\leq
  \e\varrho<\e$.
\end{thm}
\addtocounter{thm}{-1}
\begin{thm}[case $s=2$]\label{thm:exact-nn_l2}
  Given a finite set $P$ with $n$ points in $(\R^d,\ell_2)$, and a
  user-defined parameter $\e> 0$, our procedure answers exact \nn
  queries with high probability in expected $\tilde O(n^\varrho +
  n^\alpha |\ennpq{\e(2+\e)}{P}{q}|)$ time using $\tilde
  O(\frac{1}{\e}n^{2+\varrho})$ space, where $\varrho\leq
  \frac{1}{1+\e^2/(1+\e)}<1$ and $\alpha\leq
  \e\varrho<\e$.
\end{thm}

Note that in practice a trade-off must be made by the user when
choosing parameter $\e$. Indeed, the smaller $\e$, the smaller
the set $\ennpq{\e(2+\e)}{P}{q}$ and the smaller $\alpha$ compared to
$\varrho$, but on the other hand the higher $\varrho$ itself.

\begin{remark}\label{rem:exact-NN_general}
In our analysis we traded optimality for simplicity since we applied
the results from Section~\ref{sec:erpleb_Rd} verbatim. In fact, a
closer look at the problem reveals that the points of $P$ lie at least
$\dnnp{q}\geq \frac{r}{1+\e}$ away from the query point $q$ with high
probability at step 2 of the query phase. This means that no
lifting of the data into $\R^{d+1}$ is actually needed. We then
have $p'_2=p_2$, $\varrho'=\varrho$, and a careful analysis shows that
relevant choices of parameter $w$ reduce $\varrho$ down to (or at
least close to) $\frac{1}{1+\e}$. In addition and more importantly,
not having to re-embed the data means that the algorithm can be
applied in arbitrary metric spaces $(X,\dist)$ that admit
locality-sensitive families of hash functions, where the analysis
extends in a straightforward manner.
%
\end{remark}

\section{From exhaustive $r$-\pleb to exact \rnn}
\label{sec:rnn}

In this section we focus on our main problem (\rnn) and show how
it can be reduced to a single instance of $\e$-\nn search plus a
controlled number of instances of exhaustive $r$-\pleb. Although the
reduction is applicable in any metric space, we will restrict our
study to the case of $\R^d$ equipped with an $\ell_s$-norm, $s\in
(0,2]$, where the non-isometric embedding trick of
  Section~\ref{sec:erpleb_Rd} can be used to speed-up the process. The
  details of the reduction are given in Section~\ref{sec:rnn_alg}, its
  output proven correct in Section~\ref{rnn_correctness}, and its
  complexity analyzed in Section~\ref{rnn_complexity}. The reduction
  and analysis are then extended to the bichromatic setting in
  Section~\ref{sec:rnn_bichro}.  For now we begin with an overview of
  the reduction and of its key ingredients in
  Section~\ref{sec:rnn_overview}.

\subsection{Overview of the reduction}
\label{sec:rnn_overview}

Let $P$ be a finite set with $n$ points in $(\R^d,\ell_s)$, $s\in
(0,2]$.  Suppose the distance of every point $p\in P$ to its nearest
  neighbor in $P\setminus\{p\}$ has been pre-computed. Then, given a query point
  $q$, computing a solution to the \rnn query amounts to
  checking, for every point $p\in P$, whether $\dist(q,p)\leq
  \dnnpp{p}$ or $\dist(q,p)>\dnnpp{p}$: in the first case, $p$ must be
  included in the solution, whereas in the second case it must
  not. This check for point $p$ can be done by computing the solution
  $S$ of the exaustive $r$-\pleb query on input $(P,q)$, with
  $r=\dnnpp{p}$, and by including $p$ in the answer if and only if it
  belongs to~$S$. Indeed, 
\[
p\in \rnnp{q} \Leftrightarrow \dist(p,q)\leq \dnnpp{p} = r
\Leftrightarrow p\in \nnpqr{P}{q}{r} \Leftrightarrow p\in S.
\] 
Thus, computing the set $\rnnp{q}$ boils down to locating $q$ among
the set of balls $\{\ball(p, \dist(p,P)) \mid p\in P\}$.  This
observation was exploited in previous work~\cite{KM00} and serves as
the starting point of our approach.  The main problem is that the ball
radius $r$ changes with each data point $p\in P$ considered, so the
total number of exhaustive $r$-\pleb queries to be solved can be up to
linear in $n$. To reduce this number, we allow some degree of
fuzziness and use a bucketing strategy. Given a user-defined parameter
$\e>0$, at pre-processing time we compute and store $\dnnpp{p}$ for
every point $p\in P$ and then we hash the data points into buckets
according to their nearest neighbor distances, so that bucket $P_i$
contains the points $p\in P$ such that
$(1+\e)^{i-1} \leq \dnnpp{p} < (1+\e)^i$. At query time, we
  solve an exhaustive $r$-\pleb query with $r=(1+\e)^{i}$ on each
  bucket $P_i$ separately, then we consider the union $S$ of the
  solutions and prune out those points $p\in S$ such that
  $\dist(p,q)>\dnnp{p}$. Since the points $p\in P_i$ satisfy
  $(1+\e)^{i-1}\leq \dnnpp{p} < (1+\e)^{i}$, it is easily seen that
  $\rnnp{q}\subseteq S\subseteq \ernnpq{\e}{P}{q}$ and that our output
  is an admissible solution to the \rnn query.

A remaining issue is that we do not impose any constraints on
parameter $i$, so at query time we need to inspect every single
non-empty bucket $P_i$. As a result, in pathological cases such as
when all non-empty buckets are singletons, we will end up considering
a linear number of buckets, even though the set $\ernnp\e{q}$ itself
might be small or even empty. To avoid this pitfall, we limit the
range of values of $i$ to be considered thanks to the following
observations, where $y$ is an arbitrary point of $\ennpq{\e}{P}{q}$:
\begin{observation}\label{claim:closeby_pts}
  Every point $p\in\rnnp{q}$ satisfies $\dnnpp{p}\geq
  \frac{\dist(q,y)}{1+\e}$.
\end{observation}
 \begin{proof}
   Since $p\in \rnnp{q}$, we have $p\neq q$ and $\dist(p,q)\leq
   \dnnp{p}$. Moreover, since $p\neq q$ and $y\in \ennp{\e}{q}$, we
   have $\dist(q,y)\leq (1+\e)\dnnp{q}\leq (1+\e)\dist(q,p)$. It
   follows that $\dist(q,y)\leq (1+\e)\dnnp{p}$.
 \end{proof}
\begin{observation}\label{claim:far-away_pts}
  Every point $p\in\rnnp{q}$ such that $\dnnpp{p}\geq
  \frac{\dist(q,y)}{\e}$ belongs to $\ernnp{\e}{y}\cup\{y\}$.
\end{observation}
 \begin{proof}
   Since $p\in\rnnp{q}$, we have $\dist(p,q)\leq
   \dnnpp{p}$. In addition, we have $\dist(q,y)\leq \e\,\dnnpp{p}$ by
   hypothesis. Hence, $\dist(p,y)\leq \dist(p,q)+\dist(q,y)\leq
   (1+\e)\dnnpp{p}$, which means that either $p=y$ or $p\in \ernnp{\e}{y}$.
 \end{proof}
Assuming that we have precomputed a data structure that enables us to
find some $y\in\ennpq{\e}{P}{q}$, Observation~\ref{claim:closeby_pts}
ensures that we can safely ignore the buckets $P_i$ with
$i\leq\log_{1+\e} \frac{\dist(q,y)}{1+\e}$. Furthermore, assuming that
the set $\ernnp{\e}{y}$ has been precomputed,
Observation~\ref{claim:far-away_pts} ensures that the reverse
nearest neighbors of $q$ that belong to the buckets $P_i$ with
$i\geq 1+\log_{1+\e} \frac{\dist(q,y)}{\e}$ can simply be looked for
among the points of $\ernnp{\e}{y}\cup\{y\}$. Thus, the total number
of buckets to be inspected is reduced to 
$O(\frac{1}{\e}\log\frac{1}{\e})=\tilde O(\frac{1}{\e})$.

\subsection{Details of the reduction}
\label{sec:rnn_alg}

Given a finite set $P$ with $n$ points in $(\R^d,\ell_s)$, $s\in
(0,2]$, and a parameter $\e> 0$, our pre-computation phase builds a
data structure $\rnnds(P,\e)$ that stores the following pieces of
information:
\begin{slist}
\item[i.] A collection of buckets $\{P_i\}_{i\in\Z}$ that
  partition $P$. Each bucket $P_i$ contains those points $p\in P$ such
  that $(1+\e)^{i-1}\leq \dnnpp{p} < (1+\e)^i$. To fill in the
  buckets, we iterate over the points $p\in P$, we compute the
  distance $\dnnpp{p}$ exactly\footnote{This can be done either by
    brute-force or using the algorithm of
    Section~\ref{sec:exact-nn}.} and store it, and then we assign $p$
  to its corresponding bucket.
 Once this is done, the empty buckets are discarded and the non-empty
 buckets are stored in a hash table to ensure constant look-up
 time. On each non-empty bucket $P_i$ we build an $\A'(P_i, (1+\e)^i,
 \e)$ data structure using the procedure of
 Section~\ref{sec:erpleb_Rd}.
 Note that when applying Algorithm~\ref{alg:nnpre} we increase the
 number of iterations of the main loop from $\lceil c \ln |P_i|\rceil$
 to $\lceil c\ln n\rceil$, where $c=\frac{3}{\ln\frac{5}{2}}$.
\item[ii.] For each point $y\in P$, an array $P_y$ containing the
  points $p\in \ernnp{\e}{y}\cup\{y\}$, sorted by increasing
  distances $\dnnp{p}$. Building the array takes $\tilde O(n)$ time once
  $\dnnpp{p}$ has been computed for all $p\in P$.
\item[iii.] The tree $\T(P,\e)$ of
  Section~\ref{sec:renn-to-enn} and its associated $(r,\e)$-\pleb data
  structures.
\end{slist}

\smallskip

\noindent Given a point $q\in \R^d$, we answer the \rnn
query using the $\rnnds(P,\e)$ data structure as follows:
\begin{slist}
\item[1.] We use the tree $\T(P,\e)$ and its
  associated $(r,\e)$-\pleb data structures to answer an $\e$-\nn
  query, and we let $y$ be the output point.
\item[2.] We use the $\A'(P_i, (1+\e)^i, \e)$ data structure to
  answer an exhaustive $(1+\e)^i$-\pleb query on each bucket $P_i$
  separately, for $i$ lying in the range prescribed by
  Observations~\ref{claim:closeby_pts} and~\ref{claim:far-away_pts},
 and then we merge the output sets into a single set $S$.
 Note that when applying Algorithm~\ref{alg:nnquery} on $P_i$ we
 increase the number of iterations of the main loop from
 $\lceil c \ln |P_i|\rceil$ to $\lceil c \ln n\rceil$, where
 $c=\frac{3}{\ln\frac{5}{2}}$, which raises the probability of success
 of the query from $1-\frac{1}{|P_i|^2}$ (which can be as low as
 $0$ when $P_i$  is a singleton) to $1-\frac{1}{n^2}$.
\item[3.] We add to $S$ the points $p\in \ernnpq{\e}{P}{y}\cup\{y\}$ s.t.
  $\dnnp{p}\geq \frac{\dist(q,y)}{\e}$. These are found by looking up
  the value $ \frac{\dist(q,y)}{\e}$ in the sorted array $P_y$ by
  binary search, and then by iterating until the end of the array.
\item[4.] We iterate over the points $p\in S$ and remove the ones
  that do not satisfy $\dist(p,q)\leq \dnnp{p}$.
\end{slist}
\smallskip
Upon termination, we return the set $S$.  The pseudo-codes of the
preprocessing and query procedures are given in Algorithms~\ref{alg:rnnpre}
and~\ref{alg:rnnquery}.

%
\begin{algorithm}[!htb]
  \LinesNumbered
  \SetKwInOut{Input}{Input}
  \SetKwInOut{Output}{Output}

  \Input{point cloud $P\subset \R^d$, 
    parameter $\e>0$}
  \Output{$\rnnds(P,\e)$ data structure}
  \BlankLine

  Initialize $P_i := \emptyset$ for $i\in \Z$\;
  \ForEach{$p\in P$} {
    Compute $\dnnpp{p}$ exactly and store it\;
    Find $i$ s.t. $(1+\e)^{i-1}\leq \dnnpp{p}<(1+\e)^{i}$ and
    update $P_i := P_i \cup \{p\}$\;
  }
  \ForEach{$P_i\neq\emptyset$} {
    Build an $\A'(P_i, (1+\e)^{i}, \e)$ data structure\;
  }
  \ForEach{$y\in P$} {
    Build the set $\ernnp{\e}{y}\cup\{y\}$ and store it in an array $P_y$ \;
    \mbox{Sort the points $p\in P_y$ by increasing distances $\dnnp{p}$} \;
  }
  Build the tree $\T(P,\e)$ of
  Section~\ref{sec:renn-to-enn} and its associated $(r,\e)$-\pleb data
  structures \;
  \caption{\small Pre-processing phase for \rnn.}
  \label{alg:rnnpre}
\end{algorithm}
%
\begin{algorithm}[!htb]
  \LinesNumbered
  \SetKwInOut{Input}{Input}
  \SetKwInOut{Output}{Output}

  \Input{$\rnnds(P,\e)$ data structure, query point $q\in \R^d$}
  \BlankLine

Answer an $\e$-\nn query on input $(P,q)$, and let  $y$ be the output \;
  \For{$i=\left\lfloor \log_{1+\e} \frac{\bd(q,y)}{1+\e}\right\rfloor + 1
    \ \mathrm{ to }\ \left\lceil \log_{1+\e}\frac{\bd(q,y)}{\e}\right\rceil$} { 
    \If{$P_i\neq\emptyset$}{ Answer an exhaustive $(1+\e)^{i}$-\pleb query on input $(P_i,q)$, and let $S_i$ be the output \; 
    } 
  }

  Let $S:=\bigcup_i S_i$ \;

  Look up the value $\frac{\dist(q,y)}{\e}$ in the sorted array $P_y$ by binary search \;

  Iterate from the  value $\frac{\dist(q,y)}{\e}$ to the end of the array $P_y$ and insert all the visited points into $S$ \;

  \ForEach{$p\in S$ } {
    \If{$\bd(p,q) > \dnnp{p}$} {
        Remove $p$ from $S$\;
    }
  }

  Return $S$ \;
  \caption{\small Online query phase for \rnn.}
  \label{alg:rnnquery}
\end{algorithm}

\subsection{Correctness of the output}
\label{rnn_correctness}

Corollary~\ref{cor:lsh4enn} guarantees that step 1 of the query
procedure retrieves a point $y\in\ennp{\e}{q}$ with high
probability. Let us show that, given that $y\in \ennp{\e}{q}$, the
final set $S$ output by the query procedure satisfies $S=\rnnp{q}$ with high
probability.  For clarity, we let $S'$ be the set of points inserted
in $S$ at step 2 of the procedure, and $S''$ be the set of points
inserted at step 3. The output of the algorithm is then $(S'\cup
S'')\cap \rnnp{q}$.
Let $P'=\bigcup_i P_i$ for $i=\lfloor \log_{1+\e}
\frac{\bd(q,y)}{1+\e}\rfloor + 1$ to $\lceil 
\log_{1+\e}\frac{\bd(q,y)}{\e}\rceil$.
\begin{lem} \label{lem:allnnred}
  $\rnnp{q}\cap P'\subseteq S'$ with high probability.
\end{lem}
\begin{proof}
  Step 2 of the query procedure builds $S'$ by taking the union of the
  sets $S_i$ generated by answering exhaustive $(1+\e)^{i}$-\pleb
  queries on the non-empty buckets $P_i$ with query point $q$. For
  each such $P_i$, we have $\rnnpq{P}{q}\cap P_i\subseteq
  \nnpqr{P_i}{q}{(1+ \e)^i}$ since by definition every point $p\in
  \rnnpq{P}{q}\cap P_i$ satisfies $\dist(p,q)\leq \dnnp{p}\leq
  (1+\e)^i$. Now, by Theorem~\ref{thm:allnn}, we have $S_i =
  \nnpqr{P_i}{q}{(1+\e)^i}$ with probability at least
  $1-\frac{1}{n^2}$.  Thus, $\rnnp{q}\cap P_i \subseteq S_i$ with
  probability at least $1-\frac{1}{n^2}$. Since the total number of
  non-empty buckets is at most $n$, the union bound tells us that
  $\rnnp{q}\cap P'\subseteq S'$ with probability
  at least $1-\frac{1}{n}$.
\end{proof}
\begin{lem}\label{lem:rnnofann}
  Given that $y\in \ennp\e{q}$, we have $\rnnp{q}\setminus P'\subseteq
  S''$ with high probability.
\end{lem}
\begin{proof}
  The result follows from Observations~\ref{claim:closeby_pts}
  and~\ref{claim:far-away_pts}. Indeed, every point $p\in P_i$ with
  $i< \lfloor \log_{1+\e} \frac{\bd(q,y)}{1+\e}\rfloor+1$ satisfies
  $\dnnpp{p}<(1+\e)^{i}\leq\frac{\bd(q,y)}{1+\e}$ and therefore
  cannot belong to $\rnnp{q}$, by Observation~\ref{claim:closeby_pts}.
  In addition, the points $p\in \rnnp{q}\cap P_i$ with $i>\lceil
  \log_{1+\e}\frac{\bd(q,y)}{\e}\rceil$ satisfy $\dnnpp{p}\geq
  (1+\e)^{i-1}\geq \frac{\bd(q,y)}{\e}$ and therefore belong to
  $\ernnp{\e}{y}\cup\{y\}$, by Observation~\ref{claim:far-away_pts}.
  Hence, all such points $p$ are inserted in $S$ at step 3 of the
  query procedure. It follows that $\rnnp{q}\setminus P'\subseteq
  S''$.
\end{proof}

It follows from Lemmas~\ref{lem:allnnred} and~\ref{lem:rnnofann} that
$(S'\cup S'')\cap\rnnp{q}=\rnnp{q}$ with high probability.  In other
words, the set $S$ returned after step 4 of the query procedure
coincides with $\rnnp{q}$ with high probability.

\subsection{Complexity}
\label{rnn_complexity}

Let $D$ be the $s$-stable distribution used by the algorithm, and let
$p_0=\Prob(\frac{1}{1+\e})$, $p_1=\Prob(1)$, $p_2=\Prob(1+\e)$ and
$p'_2=\Prob(((1+\e)^s + (1+\e)^{-s} -1)^{1/s})$ be derived from $D$
according to Eq.~(\ref{eq:proba_collision_ls}). By
Corollary~\ref{cor:lsh4enn}, the running time of the $\e$-\nn query at
step 1 is $\tilde O(\frac{n^{\varrho}}{p_1 \ln \nicefrac{1}{p_2}})$,
where $\varrho = \frac{\ln p_1}{\ln p_2}$. Then, for $i$ ranging from
$\left\lfloor \log_{1+\e} \frac{\bd(q,y)}{1+\e}\right\rfloor + 1$ to
$\left\lceil \log_{1+\e}\frac{\bd(q,y)}{\e}\right\rceil$, the
exhaustive $(1+\e)^i$-\pleb query on the set $P_i$ takes $\tilde
O(\frac{|P_i|^{\varrho'}}{p_1}(\frac{1}{\ln \nicefrac{1}{p'_2}}+1) +
\frac{|P_i|^\alpha}{p_1} |\nnpqr{P_i}{q}{(1+\e)^{i+1}}|) = \tilde
O(\frac{n^{\varrho'}}{p_1}(\frac{1}{\ln \nicefrac{1}{p'_2}}+1) +
\frac{n^\alpha}{p_1} |\nnpqr{P_i}{q}{(1+\e)^{i+1}}|)$ time in
expectation, where $\varrho' = \frac{\ln p_1}{\ln p'_2}$ and $\alpha =
\varrho'(1-\frac{\ln p_0}{\ln p_1})$, by Theorem~\ref{thm:allnn_Rd}.
Observe that the points $p\in \nnpqr{P_i}{q}{(1+\e)^{i+1}}$ satisfy
$\dist(p,q)\leq (1+\e)^{i+1}\leq (1+\e)^2\dnnpp{p}$, so we have
$\nnpqr{P_i}{q}{(1+\e)^{i+1}}\subseteq
\ernnp{\e(2+\e)}{q}$. Furthermore, since the buckets $P_i$ are
pairwise disjoint, so are the sets $\nnpqr{P_i}{q}{(1+\e)^{i+1}}$. It
follows that the total expected time spent at step 2 is $\tilde
O(\frac{1}{\e}\frac{n^{\varrho'}}{p_1}(\frac{1}{\ln\nicefrac{1}{p'_2}}+1)
+ \frac{n^\alpha}{p_1}|\ernnp{\e(2+\e)}{q}|)$, the factor
$\frac{1}{\e}$ in the first term coming from the fact that there are
$\tilde O(\frac{1}{\e})$ iterations of the loop.  Considering now step
3, the binary search takes $O(\log_2 |P_y|)=O(\log_2 n)$ time.  For
every point $p\in P_y$ such that $\dnnpp{p}\geq
\frac{\dist(q,y)}{\e}$, we have $\dist(p,y)\geq \dnnpp{p}\geq
\frac{\dist(q,y)}{\e}$, so $\dist(p,q)\leq \dist(p,y)+\dist(y,q)\leq
(1+\e)\dist(p,y)\leq (1+\e)^2\dnnpp{p}$ since $p\in
P_y=\ernnp{\e}{y}$. It follows that $p\in\ernnp{\e(2+\e)}{q}$. Hence,
the total time spent at step 3 is $O(\log_2 n +
|\ernnp{\e(2+\e)}{q}|)$ and is therefore dominated by the time spent
at step 2.  Finally, the time spent at step 4 is dominated by the
times spent at steps 2 and 3. Combining these bounds together and
using the fact that $p'_2\geq p_2$ and $\varrho'\geq\varrho$ since
$((1+\e)^s + (1+\e)^{-s}-1)^{1/s}\leq 1+\e$, we obtain the following
query time bound (where $p'_2$ and $\varrho'$ are renamed respectively
$p_2$ and $\varrho$ for convenience):
\begin{thm}\label{thm:rnnquerytime}
Given $q\in (\R^d, \ell_s)$, the expected query time is
$\tilde O (\frac{1}{\e}\frac{n^\varrho}{p_1}(\frac{1}{\ln
  \nicefrac{1}{p_2}}+1) + \frac{n^\alpha}{p_1}
|\ernnp{\e(2+\e)}{q}|)$, where $\varrho=\frac{\ln p_1}{\ln p_2}$ and
$\alpha=\varrho(1-\frac{\ln p_0}{\ln p_1})$, the quantities
$p_0=\Prob(\frac{1}{1+\e})$, $p_1=\Prob(1)$ and
$p_2=\Prob(((1+\e)^s+(1+\e)^{-s}-1)^{1/s})$ being derived from some
$s$-stable distribution $D$ according to
Eq.~(\ref{eq:proba_collision_ls}).
\end{thm}
Replacing Theorem~\ref{thm:allnn_Rd} by its specialized versions for
$s=1$ and $s=2$ in the analysis immediately gives the following
running time bounds:
\addtocounter{thm}{-1}
\begin{thm}[case $s=1$]\label{thm:rnnquerytime_l1}
Given a query point $q\in (\R^d, \ell_1)$, the expected running time
of Algorithm~\ref{alg:rnnquery} is $\tilde O (\frac{1}{\e} n^\varrho +
n^\alpha |\ernnp{\e(2+\e)}{q}|)$,
where $\varrho\leq \frac{1}{1+\min\{\e^2,\;\sqrt{\e}\}/4}<1$ and
$\alpha\leq \e\varrho<\e$.
\end{thm}
\addtocounter{thm}{-1}
\begin{thm}[case $s=2$]\label{thm:rnnquerytime_l2}
Given a query point $q\in (\R^d, \ell_2)$, the expected running time
of Algorithm~\ref{alg:rnnquery} is $\tilde O (\frac{1}{\e} n^\varrho +
n^\alpha |\ernnp{\e(2+\e)}{q}|)$,
where $\varrho\leq \frac{1}{1+\e^2/(1+\e)}<1$ and
$\alpha\leq \e\varrho<\e$.
\end{thm}

As mentioned in Section~\ref{sec:rnn_alg}, the $\rnnds(P,\e)$ data
structure consists mainly of a collection of pairwise-disjoint
non-empty buckets, of total cardinality $n$, and for each bucket $P_i$
an $\A'(P_i, (1+\e)^{i}, \e)$ data structure of size $\tilde
O(\frac{n_i^{1+\varrho'}}{p_1})$ where $n_i=|P_i|$, by
Theorem~\ref{thm:allnn_Rd}. This gives a total size of $\tilde
O(\sum_i \frac{n_i^{1+\varrho'}}{p_1}) = \tilde
O(\frac{n^{1+\varrho'}}{p_1})$.  In addition, $\rnnds(P,\e)$ stores
the tree structure $\T(P,\e)$ and its associated $(r,\e)$-\pleb data
structures, whose total size is $\tilde
O(\frac{1}{\e}\frac{n^{1+\varrho}}{p_1})$, by
Corollary~\ref{cor:lsh4enn}. Finally, $\rnnds(P,\e)$ stores a vector
$P_y$ for each point $y\in P$, which requires a total space of $\tilde
O(\sum_{y\in P} |P_y|)$, where $|P_y|=1+|\ernnpq{\e}{P}{q}|\leq n$.
Combining these bounds and using the fact that $\varrho'\geq \varrho$,
we obtain the following bound on the size of the data structure (where
$p'_2$ and $\varrho'$ have been renamed respectively $p_2$ and
$\varrho$ for convenience):
\begin{thm}\label{thm:rnnspaceusage}
  The size of the data structure $\rnnds(P,\e)$ built by
  Algorithm~\ref{alg:rnnpre} is $\tilde O( \frac{1}{\e}
  \frac{n^{1+\varrho}}{p_1} + \sum_{y\in P}|\ernnpq{\e}{P}{y}|)= \tilde
  O(\frac{1}{\e} \frac{n^{1+\varrho}}{p_1} + n^2)$, where $\varrho=\frac{\ln
    p_1}{\ln p_2}<1$, the quantities $p_1=\Prob(1)$ and
  $p_2=\Prob(((1+\e)^s+(1+\e)^{-s}-1)^{1/s})$ being derived from some
  $s$-stable distribution $D$ according to
  Eq.~(\ref{eq:proba_collision_ls}).
\end{thm}

\subsection{Bichromatic \rnn}
\label{sec:rnn_bichro}

Let $(X,\dist)$ be a metric space, and let $B, Y$ be two finite
subsets of $X$, respectively referred to as the blue and yellow sets
in the following. Given a point $x\in X$, a {\em reverse nearest neighbor}
of $x$ in this bichromatic setting is a point $b\in B\setminus\{x\}$
such that $x\in\nnpq{Y\cup \{x\}}{b}$. Let $\rnnpq{B,Y}{x}$ denote the
set of all such points. By analogy, given a parameter $\e> 0$,
a {\em reverse $\e$-nearest neighbor of $x$} is a point $b\in
B\setminus\{x\}$ such that $x\in \ennpq\e{Y\cup\{x\}}{b}$, and let
$\ernnpq\e{B,Y}{x}$ denote the set of all such points. The bichromatic
version of Problem~\ref{problem:rnn} is
stated as follows:
\begin{problem}[Bichromatic \rnn] \label{problem:bichro-rnn}
  Given a query point $q\in X$, the {\em bichromatic reverse nearest
    neighbors query} asks to retrieve the set $\rnnpq{B,Y}{q}$.
\end{problem}
%

\begin{algorithm}[!htb]
  \LinesNumbered
  \SetKwInOut{Input}{Input}
  \SetKwInOut{Output}{Output}

  \Input{point clouds $B,Y\subset \R^d$, 
    parameter $\e>0$}
  \Output{$\rnnds(B,Y,\e)$ data structure}
  \BlankLine

  Initialize $P_i := \emptyset$ for $i\in \Z$\;
  \ForEach{$b\in B$} {
    Compute $\dnn{Y}{b}$ exactly and store it\;
    Find $i$ s.t. $(1+\e)^{i-1}\leq \dnn{Y}{b}<(1+\e)^{i}$ and
    update $P_i := P_i \cup \{b\}$\;
  }
  \ForEach{$P_i\neq\emptyset$} {
    Build an $\A'(P_i, (1+\e)^{i}, \e)$ data structure\;
  }
  \ForEach{$y\in Y$} {
    Build the set $\ernnpq{\e}{B,Y}{y}\cup(\{y\}\cap B)$ and store it 
    in an array $P_y$ \;
    Sort the points $b\in P_y$ by increasing distances $\dnn{Y}{b}$\;
  }
  Build the tree structure $\T(Y,\e)$ of
  Section~\ref{sec:renn-to-enn} and its associated $(r,\e)$-\pleb data
  structures  \;
  \caption{\small Pre-processing phase for bichromatic \rnn.}
  \label{alg:rnnpre_bichro}
\end{algorithm}
%
\begin{algorithm}[!htb]
  \LinesNumbered
  \SetKwInOut{Input}{Input}
  \SetKwInOut{Output}{Output}

  \Input{$\rnnds(B,Y,\e)$ data structure,
    query point $q\in \R^d$}
  \BlankLine

Answer an $\e$-\nn query on input $(Y,q)$, and let  $y$ be the output \;
  \For{$i=\left\lfloor \log_{1+\e} \frac{\bd(q,y)}{1+\e}\right\rfloor + 1
    \ \mathrm{ to }\ \left\lceil \log_{1+\e}\frac{\bd(q,y)}{\e}\right\rceil$ } { 
    \If{$P_i\neq\emptyset$}{ Answer an exhaustive $(1+\e)^{i}$-\pleb query on input $(P_i,q)$, and let $S_i$ be the output \; 
    } 
  }

  Let $S:=\bigcup_i S_i$\; 

  Look up the value $\frac{\dist(q,y)}{\e}$ in the sorted array $P_y$ by binary search \;

  Iterate from the value $\frac{\dist(q,y)}{\e}$ to the end of the array $P_y$ and insert the visited points into $S$ \;

  \ForEach{$b\in S$ } {
    \If{$\bd(b,q) > \dnn{Y}{b}$} {
        Remove $b$ from $S$\;
    }
  } 

  Return $S$ \;
  \caption{\small Online query phase for bichromatic \rnn.}
  \label{alg:rnnquery_bichro}
\end{algorithm}

Our strategy for answering reverse nearest neighbors queries extends
quite naturally to the bichromatic setting when the ambient space is
$\R^d$ equipped with an $\ell_s$-norm, $s\in (0,2]$. Given two finite
  subsets $B,Y$ of $\R^d$, and a parameter $\e> 0$, the data
  structure and algorithms are the same as in
  Section~\ref{sec:rnn_alg}, modulo the following minor changes:
\begin{itemize}
\item the buckets $P_i$ now partition the blue point
  set $B$, and each bucket $P_i$ gathers the points $b\in B$ such that
  $(1+\e)^{i-1}\leq \dnn{Y}{b} < (1+\e)^i$,
\item the tree structure of Section~\ref{sec:renn-to-enn} is now built
  on top of the yellow set $Y$, so we can find approximate nearest
  neighbors among the yellow points efficiently,
\item for each point $y\in Y$, we now store the set
  $\ernnpq{\e}{B,Y}{y}$ in vector $P_y$, to which we add $y$ itself
  only if the latter coincides with a point of $B$. The points in
  $P_y$ are then sorted by increasing distances to $Y$.
\end{itemize}
The details of the preprocessing and query procedures are given in
Algorithms~\ref{alg:rnnpre_bichro} and~\ref{alg:rnnquery_bichro} for
completeness. The proof of correctness with high probability and the
complexity analysis extend verbatim to the bichromatic setting, modulo
the systematic replacement of point set $P$ by either $B$ or $Y$. We
thus obtain the following guarantees:
\begin{thm}\label{thm:rnnquery_bichro}
Given a query point $q\in (\R^d, \ell_s)$,
Algorithm~\ref{alg:rnnquery_bichro} answers bichromatic \rnn
queries correctly with high probability in expected $\tilde O
(\frac{1}{\e}\frac{n^{\varrho}}{p_1}(\frac{1}{\ln
  \nicefrac{1}{p_2}}+1) + \frac{n^\alpha}{p_1}
|\ernnpq{\e(2+\e)}{B,Y}{q}|)$ time using $\tilde O(\frac{1}{\e}
\frac{n^{1+\varrho}}{p_1} + \sum_{y\in Y} |\ernnpq{\e}{B,Y}{y}|) = \tilde
O(\frac{1}{\e} \frac{n^{1+\varrho}}{p_1} + n^2)$ space, where
$n=\max\{|B|, |Y|\}$, $\varrho=\frac{\ln p_1}{\ln p_2}$ and
$\alpha=\varrho(1-\frac{\ln p_0}{\ln p_1})$, the quantities
$p_0=\Prob(\frac{1}{1+\e})$, $p_1=\Prob(1)$ and
$p_2=\Prob(((1+\e)^s+(1+\e)^{-s}-1)^{1/s})$ being derived from some
$s$-stable distribution $D$ according to
Eq.~(\ref{eq:proba_collision_ls}).
\end{thm}
\addtocounter{thm}{-1}
\begin{thm}[case $s=1$]\label{thm:rnnquery_bichro_l1}
Given a query point $q\in (\R^d, \ell_1)$,
Algorithm~\ref{alg:rnnquery_bichro} answers bichromatic \rnn queries
correctly with high probability in expected $\tilde O
(\frac{1}{\e}n^{\varrho} + n^\alpha |\ernnpq{\e(2+\e)}{B,Y}{q}|)$ time
using $\tilde O(\frac{1}{\e} n^{1+\varrho} + \sum_{y\in Y}
|\ernnpq{\e}{B,Y}{y}|) = \tilde O(\frac{1}{\e} n^{1+\varrho} + n^2)$
space, where $n=\max\{|B|, |Y|\}$, $\varrho\leq
\frac{1}{1+\min\{\e^2,\; \sqrt{\e}\}/4}<1$ and
$\alpha\leq\e\varrho<\e$.
\end{thm}
\addtocounter{thm}{-1}
\begin{thm}[case $s=2$]\label{thm:rnnquery_bichro_l2}
Given a query point $q\in (\R^d, \ell_2)$,
Algorithm~\ref{alg:rnnquery_bichro} answers bichromatic \rnn queries
correctly with high probability in expected $\tilde O
(\frac{1}{\e}n^{\varrho} + n^\alpha |\ernnpq{\e(2+\e)}{B,Y}{q}|)$ time
using $\tilde O(\frac{1}{\e} n^{1+\varrho} + \sum_{y\in Y}
|\ernnpq{\e}{B,Y}{y}|) = \tilde O(\frac{1}{\e} n^{1+\varrho} + n^2)$
space, where $n=\max\{|B|, |Y|\}$, $\varrho\leq
\frac{1}{1+\e^2/(1+\e)}<1$ and $\alpha\leq\e\varrho<\e$.
\end{thm}

\section{Conclusion}

We have introduced a novel algorithm for answering (monochromatic or
bichromatic) \rnn queries that is both provably correct and 
efficient in all dimensions. Our approach is based on a reduction of
the problem to standard $\e$-\nn search plus a controlled number of
exhaustive $r$-\pleb queries, for which we propose a speed-up of the
original LSH scheme based on a non-isometric lifting of the
data. Along the way, we obtain a new method for answering exact \nn
queries, whose complexity bounds reflect the gap in difficulty that
exists between exact and approximate queries on a given instance.

Note that the non-isometric lifting trick can be used in a more
aggressive way by applying liftings with ever more distortion, so as
to reduce the exponent $\alpha$ to arbitrarily small positive
constants. However, this comes at the price of a steady degradation of
the exponent $\varrho$, which gets closer and closer to $1$. The
question is how far up in distortion one can go before the increase of
$\varrho$ starts compensating for the reduction of $\alpha$.  Another
question in the same vein is whether $\alpha$ can be made dependent on
$n$. For instance, can $\alpha$ be reduced to $\frac{\ln\ln n}{\ln
  n}$, so the output-sensitive term in the query time depends on $\ln
n$ instead of $n^{\Theta(1)}$?  More generally, how far from the
optimal do our complexity bounds stand?

In this paper we only cared about sublinear query time and polynomial
space usage. In practice the degree of the polynomial in the space
bound matters, and in this respect the almost-cubic bound of
Theorem~\ref{thm:exact-nn} for exact \nn search is not quite
satisfactory. Moreover, the current preprocessing time may not be so good
due to the fact that some proximity sets, such as $\ernnpq{\e}{P}{y}$
in step ii of the \rnn procedure, are computed
exactly. To speed up the process one could compute them approximately,
like in previous
literature~\cite{HarPeledIndykMotwani}. Then, the outcome of the
query would likely not be exact, however it might still be
approximately correct. In other words, solving approximate \nn and
\rnn queries might help speed up the preprocessing times and
reduce the size of the data structures.

\section*{Acknowledgements}
A preliminary version of the paper was written in collaboration with
Aneesh Sharma, and some exploratory experiments were undertaken by
Maxime Br\'enon. The authors wish to acknowledge their substantial
contributions to this work. They also wish to thank Piotr Indyk for
helpful discussions about previous work in the area.